\documentclass[10pt,twocolumn,letterpaper]{article}
\usepackage{amssymb}
\usepackage{amsmath}
\usepackage{amsthm}
\usepackage[T1]{fontenc}
\usepackage{newtxmath}
\usepackage{newtxtext}

\usepackage[scaled=.8]{beramono}

\DeclareMathAlphabet\mathcal{OMS}{cmsy}{m}{n}
\SetMathAlphabet\mathcal{bold}{OMS}{cmsy}{b}{n}

\DeclareSymbolFont{AMSb}{U}{msb}{m}{n}
\DeclareSymbolFontAlphabet{\mathbb}{AMSb}

\DeclareSymbolFont{numbers}{T1}{ptm}{m}{n}
\SetSymbolFont{numbers}{bold}{T1}{ptm}{bx}{n}
\DeclareMathSymbol{0}\mathalpha{numbers}{"30}
\DeclareMathSymbol{1}\mathalpha{numbers}{"31}
\DeclareMathSymbol{2}\mathalpha{numbers}{"32}
\DeclareMathSymbol{3}\mathalpha{numbers}{"33}
\DeclareMathSymbol{4}\mathalpha{numbers}{"34}
\DeclareMathSymbol{5}\mathalpha{numbers}{"35}
\DeclareMathSymbol{6}\mathalpha{numbers}{"36}
\DeclareMathSymbol{7}\mathalpha{numbers}{"37}
\DeclareMathSymbol{8}\mathalpha{numbers}{"38}
\DeclareMathSymbol{9}\mathalpha{numbers}{"39}

\SetSymbolFont{numbers}{bold}{T1}{ptm}{b}{n}
\makeatletter
\renewcommand{\operator@font}{\mathgroup\symnumbers}
\makeatother

\DeclareRobustCommand{\%}{{\mbox{\fontencoding{\encodingdefault}\fontfamily{ptm}\selectfont
    \symbol{`\%}}}}
 \usepackage{booktabs}
\usepackage[font={small}]{caption}
\usepackage{cite}
\usepackage{color}
\usepackage{colortbl}
\usepackage{enumitem} 
\usepackage[letterpaper,left=0.65in,right=0.65in,top=0.75in,bottom=0.75in]{geometry}
\usepackage[pdfstartview=FitH,colorlinks,urlcolor=black,linkcolor=black,citecolor=black,pdfpagelabels]{hyperref}
\usepackage{framed}
\usepackage{hhline}
\PassOptionsToPackage{hyphens}{url}
\usepackage{algorithmicx, algorithm, algpseudocode}
\usepackage{caption}

\usepackage{mathtools}
\usepackage{microtype}
\usepackage{multirow}
\usepackage{pbox}
\usepackage{tikz}
\usetikzlibrary{positioning}
\usepackage{pgfmath}
\usepackage{circuitikz}
\usepackage{url}
\usepackage{xcolor}
\usepackage{xspace}

\usepackage{flushend}

\usepackage[capitalise,noabbrev]{cleveref}
\usepackage[mathscr]{euscript}		

\definecolor{LightCyan}{rgb}{0.88,1,1}

\setitemize{noitemsep,topsep=2pt,parsep=2pt,partopsep=2pt}
\setenumerate{noitemsep,topsep=2pt,parsep=2pt,partopsep=2pt}

\long\def\com#1{}

\newcommand{\itpara}[1]{\medskip\noindent\textit{#1}}

\renewcommand{\paragraph}[1]{\medskip\noindent\textbf{#1}}

\newtheorem{definition}{Definition}
\newtheorem{proposition}{Proposition}

\makeatletter

\let\c@table\c@figure
\makeatother 

\renewenvironment{thebibliography}[1]{
  \begin{OLDthebibliography}{#1}
    \setlength{\itemsep}{0em}
    \setlength{\parskip}{0em}
}
{
  \end{OLDthebibliography}
}

\newcommand{\ignore}[1]{}

\newcommand{\Hash}{\mathsf{Hash}}

\newcommand{\F}{\mathbb{F}}
\newcommand{\G}{\mathbb{G}}
\newcommand{\Z}{\mathbb{Z}}

\newcommand{\N}{\mathbb{N}}

\newcommand{\calO}{\ensuremath{\mathcal{O}}}

\newcommand{\boldone}{\ensuremath{\boldsymbol{1}}}

\newcommand{\Convert}{{\sf Convert}}
\newcommand{\Keep}{{\sf Keep}}
\newcommand{\Lose}{{\sf Lose}}

\newcommand{\zo}{\ensuremath{\{0,1\}}} \newcommand{\zon}{\ensuremath{\{0,1\}^n}} 

\theoremstyle{plain}
\newtheorem{theorem}{Theorem}[section]
\newtheorem{claim}[theorem]{Claim}
\newtheorem{lemma}[theorem]{Lemma}

\theoremstyle{definition}

\newtheorem{notation}[theorem]{Notation}

\theoremstyle{remark}
\newtheorem*{remark}{Remark}

\newcommand{\esm}[1]{\ensuremath{#1}}
\newcommand{\ms}[1]{\esm{\mathsf{#1}}}

\newcommand{\pp}{\ms{pp}}

\newcommand{\st}{\ms{st}}
\newcommand{\val}{\ms{val}}

\newcommand{\eval}{\ms{Eval}}
\newcommand{\Gen}{\ms{Gen}}

\newcommand{\getsr}{\mathrel{\mathpalette\rgetscmd\relax}}
\newcommand{\rgetscmd}{\ooalign{$\leftarrow$\cr
    \hidewidth\raisebox{1.2\height}{\scalebox{0.5}{\ \rm R}}\hidewidth\cr}}

\newcommand{\abs}[1]{\left| #1 \right|}

\newcommand{\Otilde}{\cramped{\widetilde{O}}}

\newcounter{TaskCounter}

\newcommand{\Task}[1]{\refstepcounter{TaskCounter}\paragraph{Task~\Roman{TaskCounter}: #1}}
\newcommand{\Protocol}[1]{\refstepcounter{figure}\paragraph{Protocol~\arabic{figure}: #1.}}
\crefname{proto}{Protocol}{Protocols}

\newcommand{\Algo}[1]{\refstepcounter{figure}\paragraph{Algorithm~\arabic{figure}: #1.}}
\crefname{algo}{Algorithm}{Algorithms}

\newcommand{\Eval}{{\sf Eval}}
\newcommand{\weight}{w}
\newcommand{\heavy}{H}

\newcommand{\Next}{{\sf Next}}
\newcommand{\IDPF}{{\sf IDPF}}

\newcommand{\Sim}{{\sf Sim}}

\newcommand{\Real}{{\sf Real}}
\newcommand{\Ideal}{{\sf Ideal}}
\newcommand{\Prefix}{{\sf Prefix}}

\newcommand{\name}{Poplar\xspace}
\newcommand{\Name}{Poplar\xspace}

\newcommand{\abbr}[2]{#2}

\title{Lightweight Techniques for Private Heavy Hitters}

\author{
Dan Boneh\\
\emph{Stanford}\and
Elette Boyle\\
\emph{IDC Herzliya}
\and Henry Corrigan-Gibbs\\
\emph{EPFL and MIT CSAIL}
\and Niv Gilboa\\
\emph{Ben-Gurion University}
\and Yuval Ishai\\
\emph{Technion}
}
\date{}

\begin{document}
\maketitle

\thispagestyle{plain}

\paragraph{Abstract.}
This paper presents \name, a new system for solving the
\emph{private heavy-hitters problem}.
In this problem, there are many clients and 
a small set of data-collection servers.
Each client holds a private bitstring.
The servers want to recover the set of all popular strings,
without learning anything else about any client's string.
A web-browser vendor, for instance, can use \name to figure out 
which homepages are popular, without learning any user's homepage.
We also consider the simpler {\em private subset-histogram problem},
in which the servers want to count how many clients hold strings in a
particular set without revealing this set to the clients.

\Name uses two data-collection servers and, in 
a protocol run, each client send sends only a single message
to the servers.
\Name protects client privacy against arbitrary
misbehavior by one of the servers and
our approach requires no public-key cryptography
(except for secure channels), nor general-purpose multiparty computation.
Instead, we rely on \emph{incremental distributed point functions}, 
a new cryptographic tool that allows a client to succinctly
secret-share the labels on the nodes of an
exponentially large binary tree, provided that the
tree has a single non-zero path.
Along the way, we develop new general tools for providing 
malicious security in applications of distributed point functions. 

A limitation of \name is that it reveals to the servers slightly 
more information than the set of popular strings itself. 
We precisely define and quantify this leakage and explain how to ameliorate
its effects.
In an experimental evaluation 
with two servers on opposite sides of the U.S.,
the servers can find the 200 most popular
strings among a set of 400,000 client-held 256-bit strings in 54 minutes.
Our protocols are highly parallelizable. We estimate that with 20 physical
machines per logical server, \name could compute heavy hitters over
ten million clients in just over one hour of computation.

\section{Introduction}
To improve their products, manufacturers of
hardware devices and software applications 
collect information about how their products perform in practice.
For example, when your web browser crashes today,
it prompts you to send an error report
to the vendor with the URL that triggered the crash.
For the browser-vendor, it is important to know which URLs 
are responsible for the majority of crashes.
But since these crash reports contain the URLs that you
(the user) have been visiting, sending these reports
leaks information about your browsing history to the vendor.
It takes just one subsequent data breach or
one malicious insider to expose these 
reports---and the information they contain about your browsing history---to the world.

This data-collection task is an instance of the
\emph{private heavy-hitters problem}.
In this problem, there are many clients and 
a small set of data-collection servers.
Each client holds a string (e.g., a URL that caused a browser crash).
For some threshold $t \in \N$,
the servers want to recover every string
that more than $t$ clients hold.
In this and other applications, each client's
string comes from a large universe
(the set of all URLs), so any solution that
requires enumerating over the set of all possible strings
is infeasible.

This problem comes up in an array of 
private data-collection applications:
a cellphone vendor wants to learn which mobile apps
consume the most minutes of user attention per day, without learning 
how much each person uses each app,
an electric-car company wants to learn on which roads
its cars most often run low on battery, without
learning which car was where, 
and so on.

\medskip

In this paper, we introduce \name, a system
that solves this private heavy-hitters problem
using a new suite of lightweight cryptographic techniques. 
\Name is relatively simple to implement,
is concretely efficient, unlike methods based on 
general-purpose multiparty computation~\cite{GMW87,Yao}, and 
outperforms existing approaches based on secure aggregation~\cite{Prio,melis2016}.
We expect the cryptographic tools developed in this work
to be useful in other contexts.

\Name works in the setting in which clients
communicate with two non-colluding data-collection servers.
\Name protects \emph{client privacy} as long as 
one of the two servers is honest (the other may
deviate arbitrarily from the protocol and may collude
with an unbounded number of malicious clients).
For example, the maintainer of an app store could
run one \name server and the app developer could run the other.
\Name protects \emph{correctness} against
any number of malicious clients. 
That is, the worst a malicious client can do to
disrupt the system's execution is to lie about
its own input string.

\Name requires no public-key cryptographic
operations, apart from those needed to establish
secret channels between the parties.
In terms of communication, if each client
holds an $n$-bit string and we want to achieve
$\lambda$-bit security, each client
sends a single message, of roughly $\lambda n$ 
bits, to the servers (ignoring low-order terms).
Since \name requires each client to send only
a single message to the servers, it naturally
tolerates unreliable clients: each client needs to stay online
only long enough to send its single message to the servers.
In a deployment with $C$ clients, the servers
communicate $\lambda n C$ 
bits with each other (again, ignoring low-order terms).
In terms of computation, the client invokes a
length-doubling pseudorandom generator, such as
AES in counter mode, $O(n)$ times.
When searching for strings that more than a $\tau \in (0,1]$ fraction
of clients hold, the servers perform $\approx n C/\tau$ evaluations
of a length-doubling pseudorandom generator.

To evaluate \name in practice,
we implement the end-to-end system and evaluate it on Amazon EC2 machines
on opposite sides of the U.S.
In this cross-country configuration,
we consider a set of 400,000 clients,
each holding a 256-bit string
(long enough to hold a 40-character domain name).
We configure the two servers to compute the set of heavy hitters 
held by more than 0.1\% of these clients.
The protocol between the two servers 
takes 54 minutes in total and requires under 70 KB of communication per user. 
With this parameter setting, \name concretely 
requires over $100\times$ less communication (between the servers)
and server-side computation compared to approaches based on existing cryptographic tools.

\paragraph{Our techniques.}
Our first step to solving the private heavy-hitters problem is to study
an independently useful simpler problem of computing 
\emph{private subset histograms}.
In this problem, each client holds an $n$-bit string, as before.
Now, the servers have a small set $S$ of strings
(unknown to the clients) and, for each string
$\sigma \in S$, the servers want to know how many
clients hold string $\sigma$, without learning anything else
about any client's string.
Our starting point is a simple 
protocol for this problem 
from prior work~\cite{BGI16-FSS}, in which each client sends each server a single message.
This protocol relies on the cryptographic tool of
distributed point functions~\cite{GI14,BGI15,BGI16-FSS}.
(A distributed point function is essentially a compressed
secret-sharing of a function that has
a single non-zero output.) The prior protocol~\cite{BGI16-FSS} offers a partial defense 
against malicious clients at the expense of compromising the privacy of clients against a malicious server.

Our first technical contribution is to modify this protocol to 
simultaneously protect correctness against malicious clients and achieve privacy against
a malicious server.
To do so, we develop a new lightweight malicious-secure protocol that the two servers can
run to check that they hold additive secret shares of a vector that is zero everywhere except with
a one in a single position.
Prior approaches either required additional non-colluding servers~\cite{Riposte},
did not provide malicious security~\cite{BGI16-FSS}, 
had relatively large client-to-server communication (as in Prio~\cite{Prio}),
or required additional rounds of interaction
between the clients and servers~\cite{Express}.
Applying our new building-block immediately improves the efficiency of existing
privacy-preserving systems for advertising~\cite{Adnostic} and
messaging~\cite{Dissent,Riposte,Express}.

Perhaps even more important, prior protocols~\cite{BGI16-FSS} 
do not defend against a subtle ``double-voting'' attack.
In this attack, a malicious client can cast tentative votes for
a set $S'$ of two or more strings. 
The servers only catch the cheating client if $|S'\cap S|\ge 2$,
where $S$ is the set of strings whose popularity counts the servers compute.
To prevent this kind of attack, we leverage a refined type of distributed point function
that we term \emph{extractable distributed point functions} (``extractable DPFs''). Roughly speaking, with an extractable DPF it is possible to extract from the actions of a malicious client an honest strategy that would achieve a similar effect. We show that a variant of 
the distributed-point-function construction of prior work~\cite{BGI16-FSS} is extractable in this sense when we model the underlying PRG as a random oracle.

Next, we use our protocol for private subset histograms to construct
a protocol for the $t$-heavy hitters problem.
Our approach follows that of prior work which uses subset-histograms
protocols, in the settings streaming and local-differential privacy, 
to identify heavy hitters~\cite{CKMS03,countmin,BNST17,ZKMSL20}.

In the $t$-heavy hitters problem, each client~$i$ holds 
a string $\alpha_i \in \zon$ and the servers want to learn the
set of all strings that more than $t$ clients hold, 
for some parameter $t \in \N$.
Our idea is to have the client and servers run our private subset-histogram
protocol $n$ times.
After the $\ell$th execution of the subset-histogram protocol, 
the servers learn a set $S_\ell \subseteq \zo^\ell$ 
that contains the $\ell$-bit prefix of every $t$-heavy hitter.
After~$n$ executions, the servers learn the set $S_n$ of all
$t$-heavy hitter strings.

In more detail, the clients, for their part, 
participate in $n$ executions of the 
subset-histogram protocol.
In the $\ell$th execution, for $\ell=1,\ldots,n$,
a client holding a string $\alpha \in \zon$
participates in the protocol using the prefix $\alpha|_\ell \in \zo^\ell$
as its input to the protocol, 
where $\alpha|_\ell$ is the $\ell$-bit prefix of $\alpha$.
These executions all run in parallel, so each client in fact
only sends a single message to the servers.

The servers participate in the first execution of the 
subset-histogram protocol using the set of two prefixes $S_1 = \zo$,
and learn the histogram for this set $S_1$
(i.e., the number of client strings that begin with a `0'
and the number of client strings that being with a `1').
They prune from $S_1$ all the prefixes that occur fewer than $t$ times.
Let $T_1 \subseteq S_1$ be the remaining set of prefixes.
The servers then append a `0' and a `1' to every string in $T_1$
to obtain the set $S_2 = T_1 \times \zo$. 
In the second execution of the subset-histogram protocol, 
the servers learn the histogram for the set $S_2$.
Again, they prune from $S_2$ all the elements that occur fewer than $t$ times.
Let $T_2 \subseteq S_2$ be the remaining set of prefixes.
They compute $S_3 = T_2 \times \zo$, and learn the histogram for $S_3$.
They prune $S_3$ and continue this way until finally after $n$ executions
of the subset-histogram protocol,
they obtain the set $T_n$ of all $t$-heavy hitters. 
At every step in this protocol, the size of the set $S_\ell$ is at 
most twice the size of the final answer $T_n$. 

The straightforward implementation of the above
scheme requires each client to communicate $\Omega(n^2)$ bits to each server,
where $n$ is the length of each client's private string.
This is because each client participates in $n$ instances of the private-subset-histogram
protocol, and each protocol run requires the client to send a size-$\Omega(n)$
distributed-point-function key to the servers.
Since $n \approx 256$ in our applications, the quadratic per-client
communication cost is substantial.

To reduce this cost, 
we introduce \emph{incremental distributed point functions} (``incremental DPFs''),
a new cryptographic primitive that reduces the client-to-server communication
from quadratic in the client's string length $n$ to linear in $n$.
Conceptually, this primitive gives the client a way to succinctly 
secret-share the weights on a tree that has a single path of non-zero weight in an incremental fashion.

\paragraph{Limitations.}
The main downside of \name is that it
reveals some additional---though modest and precisely
quantified---information to the data-collection
servers about the distribution of client-held strings, 
in addition to the set of heavy hitters itself. In
particular, even when an arbitrary number of
malicious clients collude with a malicious server,
this leakage depends only on the {\em multiset} of
strings held by the honest clients, without
revealing any association between clients and
strings in this multiset. Moreover, the amount of
partial information leaked about this multiset is
comparable to the output length, and only scales
logarithmically with the number of clients $C$
when the servers search for strings that
a constant fraction of clients hold. See the
precise definition of the leakage in
Section~\ref{sec:secprop}.

To protect client privacy against even this modest leakage,
we can configure \name to 
provide $\epsilon$-differential privacy~\cite{DR14}, 
in addition to its native MPC-style security properties.
The differential-privacy guarantee then ensures that
\name will never reveal ``too much'' about any client's string,
even accounting for the leakage.
To achieve $\epsilon$-differential privacy with $C$ clients, \name introduces
additive $O(1/\epsilon)$ error, compared with the larger $\Omega(\sqrt{C}/\epsilon)$ error
inherent to protocols based on randomized response~\cite{Rappor,RapporII,BS15,BNST17,BNS19}.
(\Name
provides additional privacy benefits that cannot be obtained via randomized response, such as the ability to
 securely compute on the secret-shared histogram.)

An additional limitation is that \name requires two non-colluding servers
and it does not efficiently scale to the setting of $k$
servers, tolerating $k-1$ malicious servers.
Overcoming this limitation would either require constructing
better multi-party distributed point functions~\cite{BGI16-FSS}
or using a completely different approach.

\paragraph{Contributions.}
The main contributions of this work are:
\begin{enumerate}
  \item a malicious-secure protocol for 
        private heavy hitters in the two-server setting, 
  \item a malicious-secure protocol for 
        private subset histograms in the two-server setting,
  \item the definition and construction of 
        incremental and extractable distributed point functions,
  \item a new malicious-secure protocol for
        checking that a set of parties hold shares
        of a vector of weight at most one, and
  \item implementation and evaluation of these ideas in the \name system.
\end{enumerate}

 \section{Problem statement}

\Name works in a setting in which there are 
two data-collection \emph{servers}.
\Name provides privacy as long as at least one
of these two servers executes the protocol faithfully.
(The other server may maliciously deviate from the protocol.)
There is some number $C$ of participating \emph{clients}.
Each client $i$, for $i \in \{1, \dots, C\}$, 
holds a private input string $\alpha_i \in \zon$.
The goal of the system is to allow the servers to compute
some useful aggregate statistic over the private client-held
strings $(\alpha_1, \dots, \alpha_C)$, while leaking as little 
as possible to the servers
about any individual client's string.

\itpara{Notation.}
Throughout the paper we use $\F$ to denote a prime field and $\mathbb G$ a finite Abelian group,
we use $[n]$ to denote the set of integers $\{1,\ldots,n\}$, 
and $\N$ to denote the natural numbers.
We let $\boldone\{P\}$ be the function that returns $1$ 
when the predicate $P$ is true, and returns $0$ otherwise.
We denote assignments as $x \gets 4$ and, for a finite
set $S$, the notation $x \getsr S$ indicates a uniform random draw from $S$.
For strings $a$ and $b$, $a\|b$ denotes their concatenation.

\subsection{Private-aggregation tasks}\label{sec:prob:tasks}
In this setting, there are two tasks we consider.

\Task{Subset histogram.}\label{task:subset}
In this task, the servers hold a set $S \subseteq \zon$ of strings.
For each string $\sigma \in S$, the servers want to learn the number
of clients who hold the string $\sigma$.
In some of our applications, both the clients and servers know the set $S$
(i.e., the set is public).
In other applications, the servers choose the set $S$ and may keep it
secret. 

As a concrete application, a web-browser vendor may want to use
subset histograms to privately measure the incidence of \emph{homepage hijacking}~\cite{hijack}.
A user's homepage has been ``hijacked'' if malware changes
the user's homepage browser setting without her consent.
In this application, the browser vendor has a set $S$ 
of URLs it suspects are benefiting from homepage hijacking.
The vendor wants to know, for each URL $u \in S$,
how many clients have URL $u$ as their homepage.
For this application, it is important that the
browser vendor hide the set $S$ of suspect
websites from the clients---both to avoid legal
liability and to prevent these sites from taking
evasive action.

In this application then,
each client $i$'s string $\alpha_i$ would be a representation of her
homepage URL.
The servers' set $S = \{\sigma_1, \sigma_2, \dots \}$ would be
the set of suspect URLs.
And then the output of the task would tell the browser vendor how 
many clients use each of these suspect URLs as a homepage, without
revealing to the servers which client has which homepage.

\Task{Heavy hitters.}\label{task:heavy}
In this task, the servers want to identify which strings 
are ``popular'' among the clients.
More precisely, for an integer $t \in \N$,
we say that a string $\sigma$ is a \emph{$t$-heavy hitter}
if $\sigma$ appears in the list $(\alpha_1, \dots, \alpha_C)$ more than $t$ times.
The $t$-heavy hitters task is for the servers to find all such strings. Note that, unlike the previous subset histogram task, here there is no \emph{a priori} set of candidate heavy hitters.

As an illustrative application, consider a web browser vendor
who wants to learn which URLs most crash the browser for more than
1000 clients.
Each client $i$'s string $\alpha_i$ is a representation of the last URL
its browser loaded before crashing.
The $t$-heavy hitters in the list $(\alpha_1, \dots, \alpha_C)$, for $t=1000$, reveal
to the servers which URLs crashed the browser for more than $1000$ clients.
The servers learn nothing about which client visited which URL,
nor do they learn anything about URLs that caused browser crashes for
fewer than 1000 clients.

\subsection{Communication pattern}
\label{sec:comm}

While we primarily focus on the two tasks mentioned above---subset histogram and heavy hitters---the protocols we design can be described more generally as protocols for
privately computing an aggregate statistic
$\textit{agg} = f(\alpha_1, \dots, \alpha_C)$
over the data $\alpha_1, \ldots, \alpha_C \in \zon$ of the
$C$ clients, where the function $f$ is known to the servers but
possibly not to the clients.

Because we do not allow communication between clients,
and minimal communication between the clients and the servers,
the communication pattern for the private aggregation protocol 
should be as follows:
\begin{itemize}
\item \emph{Setup:} In an optional setup phase, the servers generate 
public parameters, which they send to all clients. 
\item \emph{Upload:} The clients proceed in an arbitrary order, 
where each participating client sends a single message to Server~0
and a single message to Server~1. Alternatively, the client can send
a single message to Server~0 that includes an encryption of its second message,
which Server~0 then routes to Server~1. 
We allow no other interaction with or between the clients. 

\item \emph{Aggregate:} Servers~0 and~1 execute a protocol among themselves,
and output the resulting aggregate statistic
$\textit{agg}$.  This step may involve
multiple rounds of server-to-server
interaction.
\end{itemize}
All of the protocols that we consider in this paper and 
implement in \name obey the above communication pattern.

\subsection{Security properties}
\label{sec:secprop}

\Name is designed to provide the following security guarantees.
In \abbr{the full version of this work~\cite{full}}{\cref{app:secdefs}}, 
we provide formal security definitions.

\paragraph{Completeness:} 
If all clients and all servers honestly follow the protocol,
then the servers correctly learn $\textit{agg} = f(\alpha_1, \dots, \alpha_C)$.

\paragraph{Robustness to malicious clients:}
Informally, a malicious client cannot bias the computed aggregate statistic
$\textit{agg}$ beyond
its ability to choose its input $\alpha \in \zon$ arbitrarily.
The same should hold for a coalition of malicious clients working together, where each can cast at most  
a single vote. Whether a malicious client's vote is counted or not may depend on the set $S$ (for subset histogram) or on other client inputs (for heavy hitters).

\paragraph{{Privacy against a malicious server:}}
Informally, if one of the servers is malicious, and the other is honest,
the malicious server should learn nothing about the clients' data beyond the 
aggregate statistic $\textit{agg}$. Furthermore, even if a malicious adversary corrupts both a server and a subset of the clients, 
the adversary should learn no more than it could have learned by choosing the inputs of malicious clients and observing the output $\textit{agg}$. 

\Name's private subset-histogram protocol in Section~\ref{sec:mal} indeed meets this ideal goal, revealing to the adversary only the subset histogram of the participating honest clients. A malicious server can choose to ``disqualify'' honest clients independently of their input, so that their input does not count towards the output. (As a simple example, the server could pretend to not receive any message from a certain client.) The differentially private mechanism in Section~\ref{sec:dp} protects honest clients from being singled out via this attack. Alternatively, if too many clients are disqualified, the honest server can abort the computation.

\Name's most efficient heavy-hitters protocols in Section~\ref{sec:tree}
reveal to a malicious adversary, who corrupts one server and a subset of the clients, 
a small amount of information 
about the honest client data beyond the list of $t$-heavy hitters.  
We capture this using a \emph{leakage function} 
$L:(\zon)^C \to \zo^\ell$
that describes the extra information the adversary obtains.
A malicious server should learn nothing about the client data beyond the 
$\textit{agg}$ and $L(\alpha_1, \dots, \alpha_C)$. 
While we defer the full specification of the leakage function $L$ to
\abbr{the full version of this work~\cite{full}}{\cref{app:secdefs}}, we note here two important
features of this function: first, $L$ is {\em
symmetric} in the sense that it only depends
on the {\em multiset} of strings that
the non-disqualified honest clients hold. In
particular, the leakage reveals no association
between clients and strings in this multiset.
Second, the output length of $L$ is comparable
to that of $\textit{agg}$, and only scales
logarithmically with the number of clients $C$
when $\tau=t/C$ is fixed. Thus, \name
leaks typically much less than
a shuffling-based approach that reveals the
entire multiset. In particular, it does not often expose 
rare inputs, which are often the most sensitive.

\begin{remark}[Non threat: Correctness against malicious servers]
If one of the servers maliciously deviates from the protocol,
we do not guarantee that the other (honest) server will recover
the correct value of the aggregate statistic.
Prior private-aggregation systems offer a similarly relaxed
correctness guarantee~\cite{CSS11,Adnostic,DFKZ13,melis2016,Sepia,ARFCR10,Prio}.
In practice, \name will typically run between two organizations 
that gain no advantage by corrupting the system's output.
(In contrast, the organizations do potentially stand to benefit by learning 
the client's private data.)
So, protecting correctness is less
crucial in our setting than protecting client privacy.
Protecting correctness in the presence of malicious servers 
would be a useful extension that we leave for future work.
\end{remark}

\subsection{Alternative approaches}

We discuss a few alternative ways to solve these problems.

\paragraph{Mix-net.}
If the servers want to compute the multiset of \emph{all} client-held strings (i.e., the threshold $t=1$),
the participants can just use a two-server mix-net~\cite{C81}.
That is, each client onion encrypts her string to the two servers, who each shuffle and
decrypt the batch of strings.
Using verifiable shuffles~\cite{N01} prevents misbehavior by the servers.
In the special case of $t=1$ and with $C$ clients, this alternative has computation cost $O(C)$
(hiding polynomial factors in the security parameter), while \name would have cost $O(C^2)$.
However, without additional rounds of interaction between the clients and servers,
the mix-net-based approach does not
generalize to searching for $t$-heavy hitters with
$t > 1$, where all non $t$-heavy hitters remain hidden.
\Name does.

\paragraph{Generic MPC + ORAM.}
Another alternative solution uses general-purpose malicious-secure two-party computation 
for RAM programs~\cite{KY18,GKKKMRV12,GLO15,LO13}.
Each client sends each server an additive secret-sharing of its input string.
The servers then run a malicious-secure multiparty computation of a RAM program 
that takes as input $C$ strings (one from each client) and computes the heavy hitters.
This approach could have asymptotically optimal computational complexity $\Otilde(C + t)$, 
for heavy-hitters threshold $t$.
At the same time, multiparty computation of RAM programs---even without malicious security---is
extremely expensive in concrete terms~\cite{DS17}, as it requires implementing an oblivious RAM~\cite{GO96} client in a 
multiparty computation.
There may be more sophisticated ways to, for example, 
efficiently implement a streaming algorithm for heavy hitters~\cite{countmin}
in a multiparty computation.
We expect that such techniques will be substantially more complicated
to implement and will be concretely more expensive.

\paragraph{Counting data structures + secure aggregation.}
The count-min sketch~\cite{countmin} is a data structure used for finding
approximate heavy hitters in the context of streaming algorithms.
Melis et al.~\cite{melis2016} demonstrate that it is possible to use
secure-aggregation techniques to allow each client to anonymously insert
its input string into the data structure.
When the set of candidate heavy hitters is unknown, as in our setting,
it is possible to use a set of $n$ such counting data structures 
(where each client holds an $n$-bit string) to 
recover the heavy hitters.
The drawbacks of this approach are:
(1) 
the concrete complexity is worse than our schemes 
    since each client must send a large data-structure update message to each server (see \cref{sec:eval}),
(2) the additional leakage is substantially larger and more difficult to quantify
    than in our protocol, and
(3) these techniques only give approximate answers, where \name computes
    the heavy hitters exactly.

\paragraph{Local differential privacy.}
A beautiful line of work has considered protocols
for computing heavy hitters in the \emph{local model} 
of differential privacy, often using sophisticated variants of
randomized response~\cite{BS15,qin2016heavy,BNST17,BNS19,ZKMSL20}.
The advantage of these protocols is that they require only
a single data-collection server. In contrast, \name and others
based on multiparty computation require at least two non-colluding servers.
The downside of these protocols is that they leak a non-negligible
amount of information about each client's private string to the server.
As we describe in \cref{sec:secprop}, the leakage in \name
depends only on the multiset of private client strings.
Thus \name gives incomparably stronger privacy guarantees
and, as we discuss in \cref{sec:dp}, can also achieve differential privacy. 
In addition, when configured to provide differential privacy 
\name introduces less noise than those based on local differential privacy.
(Since we have two non-colluding servers, the noise grows essentially as it would in the
central model of differential privacy~\cite{DR14}.)

 \section{Background}

This section summarizes the existing techniques for
private aggregation that we build on in this work.

\label{sec:bg:private}

A long line of work~\cite{DFKZ13,melis2016,elahi2014privex,JJ16,Adnostic,vpriv,Sepia,ARFCR10,JK12} has constructed private-aggregation
schemes in the client/server model in which security holds as long
as the adversary cannot control all servers.
To demonstrate how these techniques work, consider the task
of computing subset histograms (Task~\ref{task:subset} of \cref{sec:prob:tasks}).
Each client $i$ holds a private string $\alpha_i \in \zon$ and
the servers hold a set $S = \{ \sigma_1, \sigma_2, \dots, \sigma_k\}$ of strings.
For each $\sigma \in S$, the servers want to know how many clients hold the string $\sigma$.

\paragraph{Distributed point functions (DPFs).}
We can use \emph{distributed point functions}~\cite{GI14,BGI15,BGI16-FSS} 
to accomplish this task in a privacy-preserving way.
A distributed point function is, at a high level, a 
technique for secret-sharing a vector of $2^n$ elements
in which only a single element is non-zero.
The important property of distributed point functions is 
that each share has only size $O(n)$, whereas a 
na\"ive secret sharing would have share size $2^n$.

More formally, a DPF scheme, parameterized by a finite field~$\F$,
consists of two routines: 
\begin{itemize}
\item $\Gen(\alpha, \beta) \to (k_0, k_1)$.
      Given a string $\alpha \in \zon$ and value $\beta \in \F$, 
      output two DPF keys representing secret shares of a
      dimension-$2^n$ vector that has value $\beta \in \F$ 
      only at the $\alpha$-th position and is zero everywhere else.
\item $\Eval(k, x) \to \F$.
      Given a DPF key $k$ and index $x \in \zon$,
      output the value of the secret-shared vector at the position 
      indexed by the string~$x$.
\end{itemize}
The DPF correctness property states that, for all strings $\alpha \in \zon$
output values $\beta \in \F$,
keys $(k_0,k_1) \gets \Gen(\alpha, \beta)$, and strings $x \in \zon$, it holds that
\[\Eval(k_0,\, x) + \Eval(k_1,\, x) = \begin{cases}
\beta &\text{if $x = \alpha$}\\
0 &\text{otherwise}
\end{cases},\]
where the addition is computed in the finite field $\F$.
Informally, the DPF security property states that
an adversary that learns either $k_0$ or $k_1$
(but not both) learns no information about the
special point $\alpha$ or its value $\beta$.

The latest DPF constructions~\cite{BGI16-FSS}, on
a domain of size $2^n$, have keys of length
roughly $\lambda n$ bits, when instantiated with
a length-doubling PRG that uses $\lambda$-bit
keys.

\paragraph{A simple protocol for private subset histograms.}
Given DPFs, we can solve the subset-histogram problem using the following
simple protocol, which we illustrate in \cref{proto:subset}.
At a high level, each client $i$ uses DPFs to create a secret sharing of a
vector of dimension $2^n$. This vector is zero everywhere except that it has
``$1$'' at the position indexed by client $i$'s input string $\alpha_i \in \zon$.
To learn how many clients hold a particular string $\sigma$, the servers
can compute, for each client $i$, the shares of the $\sigma$-th value in the
$i$th client's secret-shared vector.
By publishing the sum of these shares, the servers learn exactly how many
clients held string $\sigma$.

\begin{figure}
\begin{framed}
\Protocol{Private subset histograms} \label[proto]{proto:subset}
There are two servers and $C$ clients.
Each client $i$, for $i \in [C]$ holds a string $\alpha_i \in \zon$.
The servers hold a set $S \subseteq \zon$ of strings.
For each string $\sigma \in S$, the servers want to learn the
number of clients who hold $\sigma$.
The protocol uses a 
prime field $\F$ with $\abs{\F} > C$.

The protocol is as follows:
\begin{enumerate}
  \item \label{step:submit} Each client $i \in \{1, \dots, C\}$, 
    on input string $\alpha_i \in \zon$ prepares a pair of 
      DPF keys as $(k_{i0}, k_{i1}) \gets \Gen(\alpha_i, 1)$.
      The client sends $k_{i0}$ to server $0$ and $k_{i1}$ to server $1$.
\item For each string $\sigma_j \in S$, each server $b \in \zo$ 
      computes the sum of its DPF keys evaluated at the string $\sigma_j$:
      \[ \val_{jb} \gets \sum_{i=1}^C \Eval(k_{ib}, \sigma_j) \quad \in \F.\]
      Each server $\beta \in \zo$ then publishes the values
      \[ (\val_{1b}, \dots, \val_{\abs{S} b}) \quad \in \F^{\abs{S}}. \]
\item Finally, for each string $\sigma_j \in S$, each server can conclude
  that the number of clients who hold string $\sigma_j$ is
      $\val_{j0} + \val_{j1} \in \F$. 
\end{enumerate}
\end{framed}
\end{figure}

Correctness holds since
\begin{align*}
  \val_{j0} + \val_{j1} &= \sum_{i=1}^C \Eval(k_{i0}, \sigma_j)
    + \sum_{i=1}^C \Eval(k_{i1}, \sigma_j)\\
  &= \sum_{i=1}^C \left ( \Eval(k_{i0}, \sigma_j) + \Eval(k_{1i}, \sigma_j) \right)\\
  &= \sum_{i=1}^C \boldone\{ \alpha_i = \sigma_j \},
\end{align*}
which is exactly the number of clients who hold string $\sigma_j$.

As long as one of the two servers is honest, a
fully malicious adversary controlling the other server
and any number of clients learns nothing about the 
honest clients' inputs, apart from what the subset histogram
itself leaks.

\medskip

In the following sections, we show how to extend this simple
scheme to protect against corruption attacks by malicious
clients (\cref{sec:mal}) and support computing heavy hitters (\cref{sec:tree} and~\ref{sec:inc}).
In \cref{sec:dp}, we demonstrate that it is possible 
to achieve user-level differential privacy with these methods as well. Finally, in
\cref{sec:eval} we provide an experimental evaluation of the efficiency of the heavy-hitters protocol.

\section{Privacy-preserving subset histograms\\ via malicious-secure sketching}
\label{sec:mal}

In this section, we show how to modify the simple
scheme of \cref{sec:bg:private} to protect against
corruption attacks by malicious clients.

In the scheme of \cref{sec:bg:private},
if even \emph{one} of the participating clients is malicious,
it can completely corrupt the histogram that the servers recover.
In particular, in Step~\ref{step:submit} of the protocol above,
a malicious client can send malformed DPF keys to the servers.
A client who mounts this attack can prevent the servers from 
recovering any output (i.e., the servers get only pseudorandom garbage)
or can manipulate the statistics (i.e., the client can arbitrarily influence
the histogram the servers recover).

For example, if the servers are using this private-subset-histogram scheme
to measure the incidence of homepage hijacking (cf. \cref{sec:prob:tasks}),
a single malicious client could manipulate the output histogram to make it
look as if no homepage hijacking was taking place.

\subsection{Prior work: Sketching for malicious clients}
\label{sec:mal:prior}
Prior work~\cite{flpcp,Express,Riposte,BGI16-FSS} has presented techniques to harden
the simple scheme of \cref{sec:bg:private}
against misbehavior by malicious clients.
These approaches use similar methods:
before the servers accept the pair of DPF keys from the client,
the servers check that the DPF keys are ``well formed.''
That is, the two servers check that the DPF keys submitted
by each client expand to shares of a vector that is zero everywhere
and one at a single position.

More specifically, given a pair of client-submitted DPF keys $(k_0, k_1)$,
each server $b \in \zo$ evaluates its DPF key $k_b$ on each 
element of the set $S = \{\sigma_1, \sigma_2, \dots \}$ to produce a vector
\[ \bar{v}_b = \big( \Eval(k_b, \sigma_1),\, \dots,\, \Eval(k_b, \sigma_{\abs{S}}) \big) \quad \in \F^{\abs{S}}.\]
Say that $\bar v = \bar v_0 + \bar v_1 \in \F^{\abs{S}}$ is ``valid''
if it zero everywhere with a one at a single index (and is ``invalid'' otherwise).
The servers then run a ``sketching'' protocol to check that 
$\bar v$ is valid.

The protocol should be:
\begin{itemize}
  \item \textbf{Complete.} If $\bar v_0 + \bar v_1$ is valid, the servers always accept.
  \item \textbf{Sound.} If $\bar v_0 + \bar v_1$ is invalid, the servers reject almost always.
  \item \textbf{Zero knowledge.} A single malicious server ``learns nothing'' by running the protocol, 
              apart from the fact that $\bar v_0 + \bar v_1$ is valid.
              In particular, the malicious server does not learn the location or value of
              the non-zero element.
              We can use a simulation-based definition to formalize this security property.
\end{itemize}

Existing sketching techniques suffer from two shortcomings:
\begin{itemize}
  \item \emph{No protection against malicious servers.}
        Existing sketching protocols for checking that the secret-shared vector $\bar v$
        has weight one either 
        do not protect client privacy against malicious behavior by the servers~\cite{BGI16-FSS}.
        (Techniques that do protect against malicious servers, 
        either have client-to-server communication that grows linearly in the length of the 
        vector being checked, as in Prio~\cite{Prio},
        or require extra rounds of interaction between the servers and client~\cite{Express,flpcp}, or 
        require extra non-colluding servers~\cite{Riposte,Blinder}.)

  \item \emph{Weak protection against malicious clients.}
        A more fundamental---and more subtle---problem in our setting
        is that these sketching methods do not necessarily
        prevent a malicious client from influencing the output more than it should,
        as prior work observes~\cite{BGI16-FSS}.

        As an extreme example, say that the servers' set $S$ consists of a single string $\sigma$
        that is unknown to the clients.
        An honest client will submit a pair of DPF keys that expand to shares of
        a vector that contains a one at a \emph{single} coordinate.
        In contrast, a malicious client can submit a pair of DPF keys that expand
        to shares of a vector that is one at \emph{every} coordinate.
        Even if the servers check that their keys expand to shares of a vector
        of weight one in the singleton set $S$, the servers will not detect this attack.

        In this way, the malicious client can have more influence on the output
        than honest clients do.
\end{itemize}

\subsection{New tool: Malicious-secure sketching}
\label{sec:mal:sketch}
Our first contribution of this section is to give a new lightweight
protocol that allows the servers to check that they are holding
additive shares $\bar v_0$ and $\bar v_1$ of 
a vector $\bar v \in \zo^m \subseteq \F^m$ of weight one 
(i.e., that has a single non-zero entry), where $\F$ is a prime field.
Unlike prior approaches, we protect against malicious misbehavior by either of
the two servers, without needing extra interaction with the client
and without needing extra servers.

Our idea is to modify a 
sketching
protocol of Boyle et al.~\cite{BGI16-FSS} (with security against semi-honest servers) to protect it against
malicious behavior on the part of the servers.
To do so, we have the client encode its vector $\bar v$ using a 
redundant, ``authenticated'' randomized encoding, inspired by techniques
from the literature on malicious-secure multiparty computation~\cite{AMD08,SPDZ}. 
We construct the encoding in such a way that if
either server tampers with the client's vector,
the honest servers will reject the client's vector 
with overwhelming probability. Simultaneously protecting against both malicious clients and a malicious server while minimizing the extra overhead is a delicate balancing act, we discuss below.

\paragraph{Encoding.}
In our scheme, we have the client choose a random value $\kappa \getsr \F$
and then encode its vector $\bar v \in \F^m$ as the pair:
$(\bar v,\, \kappa \bar v) \in \F^m \times \F^m$.
In words: the encoding consists of (a) the vector $\bar v$ and (b) the 
vector $\bar v$ scaled by a random value $\kappa \in \F$.
The client sends an additive share of this 
pair $(\bar v, \kappa \bar v)$ to each of the two servers.
Since $\bar v$ has weight one, both $\bar v$ and $\kappa \bar v$ 
are non-zero only at the same single coordinate.
The client can then represent each share of this 
tuple using a single DPF instance with a longer payload.

The client also provides the servers with some correlated 
randomness, as we discuss below, which the servers use
to run a two-party secure computation.

\paragraph{Sketching.}
The servers receive from the client additive shares of
a tuple $(\bar v, {\bar v}^*)$.
If the client is honest then $\bar v^* = \kappa \bar v$.

As in the protocol from~\cite{BGI16-FSS}, the servers then 
jointly sample a uniform random vector
$\bar r = (r_1, \dots, r_m) \in \F^m$ and compute
$\bar r^* = (r_1^2, \dots, r_m^2) \in \F^m$.
(The servers could generate the random vector $\bar r$ using a 
pseudorandom generator, such as AES in counter mode, 
seeded with a shared secret. Or, for information-theoretic
security when $\abs{\F}$ is large, the servers
could take $\bar r = (r, r^2, r^3, \dots, r^m)$.)

Now, the servers compute the inner product of these sketch vectors
with both the client's data vector $\bar v$ \emph{and} their shares of
the encoded vector ${\bar v}^*$.
That is, for $b \in \zo$, server $b$ computes:
\begin{align*}
  z_b &\gets \langle \bar r, \bar v_b \rangle \in \F;& z^*_b &\gets \langle \bar r^*, \bar v_b \rangle \in \F;
  &z^{**}_b &\gets \langle \bar r,   \bar v_b^* \rangle \in \F.
\end{align*}

\paragraph{Decision.}
Finally, the servers use a constant-size 
secure computation
to check that the original sketch would
have accepted.
Letting 
$z \gets z_0 + z_1$, 
$z^* \gets z^*_0 + z^*_1$, and
$z^{**} \gets z^{**}_0 + z^{**}_1$,
the servers use secure computation to evaluate:
	\begin{equation} \label{eq:verify}
	(z^2 - z^*) + (\kappa \cdot z - z^{**}) \in \F
	\end{equation}
and check that the output is 0. Note that the first term corresponds to the original sketch verification of~\cite{BGI16-FSS}, and the second term corresponds to checking consistency of the sharing $(\bar v,\kappa \bar v)$.

Intuitively, the second, $\kappa \bar v$-computed term will play a protecting role in the servers' verification polynomial: any attempt of a malicious server to launch a conditional failure attack by modifying the sketch $z$ to $(z + \Delta)$ will result in masking the nonzero (possibly sensitive) contribution of the first term by random garbage in the second term, from the corresponding $\kappa \Delta$ term of $\kappa (z + \Delta)$.

We remark that the function (\ref{eq:verify}) on inputs $z,z^*,z^{**}$ as written is not publicly known to the servers, due to the secret client-selected $\kappa$ term. A natural approach is to provide the servers additionally with secret shares of $\kappa$, to be treated as a further input.\footnote{This approach indeed will work, though requires care to address the servers' ability to provide additive offsets to $\kappa$. Our implementation uses a protocol based on this approach, which is slightly less efficient than the one presented here.}
Instead, we provide a direct approach for the client to enable secure computation of (\ref{eq:verify}) via appropriate correlated randomness. 

The idea follows the general approach of Boyle et al.~\cite{BGI19}, extending Beaver's notion of multiplication triples~\cite{B91} to more general functions including polynomial evaluation. Here, the client will provide the servers with additive secret shares of random offsets $a,b,c$, which they will use to publish masked inputs $Z \gets (z+a)$, $Z^* \gets (z^*+b)$, and $Z^{**} \gets (z^{**}+c)$. Then, in addition, the client will provide secret shares of each coefficient in the resulting polynomial that they wish to compute:
	\begin{align*} [(Z&-a)^2 - (Z^*-b)] + [\kappa \cdot (Z - a) - (Z^{**}-c) ] \\
&= Z^2 + Z^* - Z^{**} + Z[-2a + \kappa] + [a^2+b - a\kappa + c].
	\end{align*}
That is, the client will give additive secret shares of $A:= [-2a + \kappa]$ and $B := [a^2+b - a\kappa + c]$. To evaluate, the servers each apply the above polynomial on the publicly known values $Z,Z^*,Z^{**}$, using their share of each coefficient; this results in additive shares of the desired output.

\paragraph{Security.} 
Given an honest client, the client-aided two-party computation protocol provides security against a malicious server, up to additive attacks on the inputs $z,z^*,z^{**}$ and output of the computation. The latter is irrelevant in regard to client privacy (recall we do not address correctness in the face of a malicious server). As mentioned above, any additive attack on the inputs $(z + \Delta),~(z^* + \Delta^*),~(z^{**}+\Delta^{**})$ will result in either random garbage output (if $\Delta \neq 0$) or server-predictable output (if $\Delta = 0$).

At the same time, the protocol preserves security against a malicious client. A malicious client has the ability to send invalid values for $\bar v, \bar v^*$ (supposedly $\kappa \bar v$), $A,B$. However, incorporating these malicious values into the expression evaluated by the servers still results in an analogous polynomial in the servers' secret values $r_1,\dots,r_m$ as in~\cite{BGI16-FSS}, and application of Schwartz-Zippel similarly implies that any invalid choice of $\bar v$ will result in nonzero output evaluation with probability $1-2/|\F|$.

\paragraph{Complexity.} Altogether, the client must provide: DPF shares of $(\bar v,\kappa \bar v)$, and additive shares of $a,b,c,A,B \in \F$. Since the desired values of $a,b,c$ are independent random field elements, these shares can be compressed (also across levels of the tree) using PRG seeds, which amortizes away their required communication. This results in extra (amortized) $3\log|\F|$ bits sent to each server, coming from the increased DPF key size (extra $\F$ element for $\kappa$-multiplied payload) plus shares of 2 field elements~$(A,B)$.

For the sketch verification, the servers must exchange masked input shares of $z,z^*,z^{**}$ in the first round, and then shares of the computed output in a second round. This corresponds to $4\log |\F|$ bits of communication of each server to the other, split across two rounds.

\medskip 
We provide a more complete treatment of the sketching procedure in \abbr{the full version of this work~\cite{full}}{\cref{app:mal-sketch}}.

\subsection{New tool: Extractable DPFs} \label{sec:mal:extract}

As discussed in \cref{sec:mal:prior}, 
there is a second shortcoming to using sketching-based techniques
to protect against malicious clients in our setting.
The problem is that if the servers only sketch the client-provided
DPF keys on the strings in the subset $S$, a cheating client can potentially
gain undue influence by having its DPF keys evaluate to $1$ on many different
strings in $\zon$.
The client will evade detection as long as the client's keys evaluate to $1$
on only a single point in the subset $S$.

We address this second problem by giving a refined analysis of our
DPF construction, which is based on the
state-of-the-art DPF construction of~\cite{BGI16-FSS}.
In that construction, each DPF key has a ``public part''---which is identical
for both keys---and a ``private part''---which differs between the two DPF keys.
\abbr{In the full version of this paper~\cite{full}, we prove}
{We show (\cref{lm:taint})} that using this DPF construction,
when instantiated in the random-oracle model, 
and with a large output space,
it is computationally infeasible for a client 
to find malformed DPF keys that 
(a) have the same public part and
(b) represent the sharing of a vector that is $1$ at more than one position known to the 
client.
Moreover, it is possible to efficiently extract the position of $1$ from the oracle
queries made by a malicious client.
We term this strengthened type of DPF an ``extractable DPF.''

This gives the servers a way to check for client misbehavior: the
servers can just check that their DPF keys have identical public parts
and then conclude that the keys must represent shares of a vector that
contains a ``$1$'' at a single relevant index, at most.

\paragraph{The technical idea.}
Working in the random-oracle model~\cite{BR93}, where the underlying PRG is 
a truly random function, we
show that any cheating strategy by a client in
$\Gen$ is restricted in the following sense. Let $k_0,k_1$ denote the private parts of DPF keys
and $\pp$ the public part.  
With
high probability, a malicious client that
generates DPF keys $(k_0^*,k_1^*,\pp^*)$, and is
limited in the number of calls it makes to the
random oracle, can find at most one string
$x$ such that $\Eval(k_0^*,\pp^*, x)+\Eval(k_1^*,\pp^*, x) = 1$.
In contrast, the client can easily generate
keys and multiple strings $x$ such that
$\Eval(k_0^*,\pp^*, x)+\Eval(k_1^*,\pp^*, x)=0$, as in
a valid key, or $\Eval(k_0^*,\pp^*, x)+\Eval(k_1^*,\pp^*, x)$ is
a random value in the (large) output space. 
However, finding two pairs of keys whose outputs evaluate 
to ``$1$'' in two different known locations is infeasible.
Intuitively, the structure of the
keys enables the client to fully control a non-zero
value at only one location $x$. 

When used in combination with the sketching approach of \cref{sec:mal:sketch},
this fact essentially implies a complete defense against malicious clients. Indeed, 
uniqueness of the ``$1$'' location means that only this specific vote can be counted,
since other nonzero locations will 
either be caught by the sketching or will not be part of $S$ and therefore
not influence the output.

\medskip

Overall, combining the malicious-secure sketching technique
of \cref{sec:mal:sketch} with extractable DPFs 
gives a protocol for private subset histograms that
defends privacy against a malicious server and
correctness against a malicious client.
We note that a similar combination can be useful for other applications of DPF in which the DPF is only evaluated on a strict subset of the input domain. Such applications include private information retrieval by keywords, private distributed storage, and more~\cite{BGI16-FSS}. 

The following definition formalizes this notion of extractable DPF in the random-oracle model. Since we envision other applications, we consider here a general (Abelian) output group $\mathbb G$, rather than a finite field $\F$.
Syntactically, an \emph{extractable DPF} scheme is a DPF scheme $(\Gen,\Eval)$ with the modification 
that the $\Gen$ algorithm has an additional output $\pp$ (public parameters) 
that the $\Eval$ algorithm takes
as an additional input. Our analysis assumes that the input length $n$, group $\G$, and target nonzero payload $\beta^*$ ($\beta^*=1$ by default) are chosen 
independently of the random oracle.

\begin{definition}[Extractable DPF, Simplified]
\label{def-edpf-simple}
We say that a DPF scheme in the random-oracle model is {\em extractable} if there is an efficient extractor $E$,
such that every efficient adversary $A$ wins the following game with negligible probability in the security parameter $\lambda$, 
taken over
the choice of a random oracle~$G$ and the secret random coins of $A$.
\begin{itemize}
  \item $(1^n,\G,\beta^*) \leftarrow A(1^\lambda)$, where $\G$ is an Abelian group of size $|{\mathbb G}|\ge 2^\lambda$ and $\beta^*$ is a nonzero group element.
\item $(k^*_0,k^*_1,\pp^*,x^*) \leftarrow A^G(1^\lambda,1^n,\G,\beta^*)$, where $x^* \in \zon$, and $G$ is a  
random oracle. We assume that $\pp^*$ includes the public values $(1^\lambda,1^n,\G)$. 
\item $x \leftarrow E(k^*_0,k^*_1,\pp^*,\beta^*,T)$, where $x \in \zon$ and $T=\{ q_1,\ldots,q_t \}$ is the transcript of $A$'s $t$ oracle queries. 

\end{itemize}
We say that $A$ wins the game if $x^* \neq x$ and $\Eval^G(k_0^*,\pp^*, x^*) +\Eval^G(k_1^*,\pp^*, x^*) = \beta^*$.
\end{definition}

Note that in the above definition, the goal of the extractor $E$ is to find the only input $x$ known to $A$ on which the output is~$\beta^*$. If $A$ could find two or more such inputs, it could win the game with high probability by picking $x^*$ at random from this list. 
In \abbr{the full version of this work~\cite{full}}{\cref{app:extractx}}, we define a more general notion of
extractability, which applies to incremental DPFs (\cref{sec:inc}) and prove the following claim.
\begin{lemma}[Informal]\label{lem:esimple}
The public-parameter variant of the DPF from~\cite{BGI16-FSS} is an extractable DPF with winning probability bounded by
$\epsilon_A = \left(4t^2+2nt+1 \right)/2^\lambda$, where $n$ and $t$ are the length of $x^*$ output by $A$ and number of oracle calls made by $A$, respectively, and  $\lambda$ is the security parameter. The same holds for the Incremental DPF we construct in \cref{sec:inc}.
\end{lemma}

\section{Private heavy hitters} 
\label{sec:tree}

We now turn to the problem of collecting
\emph{$t$-heavy hitters} in a privacy-preserving way
(Task~\ref{task:heavy} of \cref{sec:prob:tasks}).
As before, there are~$C$ clients and
each client $i$ holds a string $\alpha_i \in \zon$.
Now, for a parameter $t \in \N$, 
the servers want to learn every string that appears
in the list $(\alpha_1, \dots, \alpha_C)$ 
at least $t$ times. 

We first show in \cref{sec:tree:prefix},
following prior work~\cite{CKMS03,countmin,BNST17,ZKMSL20},
that the servers can efficiently find all
$t$-heavy hitters by making what we
call ``prefix-count queries'' to the 
list of client strings $(\alpha_1, \dots, \alpha_C)$.
Next, in \cref{sec:tree:alldpf}, we show how each client $i$ can 
give the servers a secret-shared encoding of its string $\alpha_i$
that enables the servers to very efficiently make prefix-count
queries to the list of client strings $(\alpha_1, \dots, \alpha_C)$.

\begin{figure}
  \centering
  \usetikzlibrary{patterns}
\usetikzlibrary{backgrounds,intersections,calc}

\usetikzlibrary{calc}

\definecolor{few-gray-bright}{HTML}{010202}
\definecolor{few-red-bright}{HTML}{EE2E2F}
\definecolor{few-green-bright}{HTML}{008C48}
\definecolor{few-blue-bright}{HTML}{185AA9}
\definecolor{few-orange-bright}{HTML}{F47D23}
\definecolor{few-purple-bright}{HTML}{662C91}
\definecolor{few-brown-bright}{HTML}{A21D21}
\definecolor{few-pink-bright}{HTML}{B43894}
\definecolor{few-gray}{HTML}{737373}
\definecolor{few-red}{HTML}{F15A60}
\definecolor{few-green}{HTML}{7AC36A}
\definecolor{few-blue}{HTML}{5A9BD4}
\definecolor{few-orange}{HTML}{FAA75B}
\definecolor{few-purple}{HTML}{9E67AB}
\definecolor{few-brown}{HTML}{CE7058}
\definecolor{few-pink}{HTML}{D77FB4}
\definecolor{few-gray-light}{HTML}{CCCCCC}
\definecolor{few-red-light}{HTML}{F2AFAD}
\definecolor{few-green-light}{HTML}{D9E4AA}
\definecolor{few-blue-light}{HTML}{B8D2EC}
\definecolor{few-orange-light}{HTML}{F3D1B0}
\definecolor{few-purple-light}{HTML}{D5B2D4}
\definecolor{few-brown-light}{HTML}{DDB9A9}
\definecolor{few-pink-light}{HTML}{EBC0DA}

\colorlet{cryptcolor}{few-purple-bright}
\colorlet{cryptcolorp}{few-green-bright}

\tikzset{
  onpath/.style={fill=red!20,draw=black,text=black},
  highlight/.style={fill=yellow!30, rounded corners},
  offpath/.style={fill=black!15, text=black!70, draw=black!30
  },
  empty/.style={fill=black!20,draw=none,text=black!70},
  toptop/.style={draw=none,rounded corners,
    minimum width=3ex, minimum height=3ex,
    fill=none
    },
  crypt/.style={draw=cryptcolor,thick,
    minimum width=3ex, minimum height=3.5ex,
    fill=white},
  link/.style={draw,thick,->},
}

\begin{tikzpicture}[scale=0.9, transform shape]
\foreach \side in {L} {

\newcommand{\topstyle}{toptop}

\node[crypt] (n) at (4.5,-1.5) {$w_\epsilon = \textbf{3}$};
\node[crypt] (n0) at (2,-2) {$w_0 = \textbf{1}$};
\node[crypt] (n1) at (7,-2) {$w_1 = \textbf{2}$};

\foreach \d/\v in {1/00,2/01,3/10,4/11} {
  \node[crypt,\if\d2empty\fi] at (2.5*\d-1.75,-3) (data\d) {$w_{\v} = \textbf{\if\d 11\fi\if\d20\fi\if\d31\fi\if\d41\fi}$};
}

\newcounter{alphacount}
\newcommand{\drawalpha}{\node[highlight,
  ] at (0.25+1.25*\d-1.5,-4.5) {\stepcounter{alphacount}$\alpha_{\arabic{alphacount}} = {\texttt{\v}}$};}
\foreach \d/\v/\sty/\weight in {1/000/empty/0,2/001//1,3/010/empty/0,4/011/empty/0,5/100/empty/0,6/101//1,7/110/empty/0,8/111//1}
  {

  \if\d 2\drawalpha{}\fi
  \if\d 6\drawalpha{}\fi
  \if\d 8\drawalpha{}\fi

  \node[crypt,\sty] at (0.25+1.25*\d-1.5,-4)
  (leaf\d) 
  {\tiny $w_{\v} = \textbf{\weight}$};

}
\foreach \d in {0,1} {
  \draw[link] (n) -- (n\d);
}
\draw[link] (n0) -- (data1);
\draw[link] (n0) -- (data2);
\draw[link] (n1) -- (data3);
\draw[link] (n1) -- (data4);

\draw[link] (data1)-- (leaf1);
\draw[link] (data1)-- (leaf2);
\draw[link] (data2)-- (leaf3);
\draw[link] (data2)-- (leaf4);
\draw[link] (data3)-- (leaf5);
\draw[link] (data3)-- (leaf6);
\draw[link] (data4)-- (leaf7);
\draw[link] (data4)-- (leaf8);

}

\end{tikzpicture}
   \vspace{-12pt}
  \caption{An example prefix tree on strings $(\alpha_1, \alpha_2, \alpha_3)$ of length $n=3$.
          The weight $w_p$ of a prefix $p \in \zo^*$ is the number of strings in the tree
          that have $p$ as a prefix.}
  \label{fig:search}
\end{figure}

The resulting protocol is lightweight: the client sends
roughly $n$ PRG keys to each server.
When configured to search for $t$-heavy hitters for $t=\tau C$, the protocol
requires server-to-server communication $O(\lambda nC/\tau)$ and server-to-server
computation dominated by $O(nC/\tau)$ PRG operations. The protocol requires $O(n)$ rounds
of communication.

\subsection{Heavy hitters via prefix-count queries}
\label{sec:tree:prefix}

As a first step to understand our approach, imagine that,
for any string $p \in \zo^*$, 
the servers can make queries of the form:
\begin{quote}
How many of the clients' input strings
$\alpha_1, \dots, \alpha_C \in \zon$ start with the prefix $p$?
\end{quote}
We call these ``prefix-count queries.''
For example, suppose there are three clients with strings 
$(\alpha_1, \alpha_2, \alpha_3) = (\texttt{001}, \texttt{101}, \texttt{111})$.
The answer to the query ``$p=\epsilon$'' (where $\epsilon$ is the empty string) 
would be ``$3$,''
the answer to the query ``$p=\texttt{1}$'' would be ``$2$,''
the answer to the query ``$p=\texttt{10}$'' would be ``$1$,''
the answer to the query ``$p = \texttt{101}$'' would be ``$1$,'' and 
the answer to the query ``$p = \texttt{01}$'' would be ``$0$.''

We first show that if the servers can 
get the answers to such queries, then
they can use a simple algorithm to
efficiently enumerate all $t$-heavy hitters
among the list of all clients' input strings.
This is a classic observation from the literature
on streaming algorithms for heavy hitters~\cite{CKMS03,countmin},
which also appears in recent work on heavy hitters in the local 
model of differential privacy~\cite{BNST17} and
in the context of federated learning~\cite{ZKMSL20}.

This algorithm corresponds to a breadth-first-search of the
prefix tree corresponding to the set of strings (\cref{fig:search}),
in which the search algorithm prunes nodes of weight less than $t$.
To give some intuition for how the algorithm works:
let us say that a prefix string $p \in \zo^*$ is a ``heavy'' if at 
least $t$ strings in $\alpha_1, \dots, \alpha_C \in \zon$ start with $p$.
Then we have the following observations:
\begin{itemize}
  \item The empty string $\epsilon$ is always heavy.
  \item If a string $p$ is not heavy, then 
        $p \| \texttt{0}$ and $p \| \texttt{1}$ are not heavy.
  \item If a string $p$ is heavy and $p$ is $n$ characters
        long (i.e., $\abs{p} = n$), then $p$ is a $t$-heavy hitter. 
\end{itemize}
These three observations immediately give rise to \cref{proto:prefix}.
For each prefix length $\ell \in \{0, \dots, n\}$, we construct
the set $\heavy_\ell$ of heavy strings of length $\ell$.
The set $\heavy_0$ consists of the empty string $\epsilon$,
since $\epsilon$ is always heavy 
(assuming, without loss of generality that $t \leq C$).
We construct the set $\heavy_\ell$ by appending \texttt{0}
and \texttt{1} to each element of $\heavy_{\ell-1}$ and checking
whether the resulting string is heavy.
And finally, $\heavy_n$ consists of all of the $t$-heavy hitters.

\begin{figure}
\begin{framed}
{\small
\Algo{$t$-heavy hitters from prefix-count queries}\label[algo]{proto:prefix}
The algorithm is parameterized by a string length $n \in \zo$ and 
a threshold $t \in \N$.

\paragraph{Input:} The algorithm has no explicit input,
  but it has access to a ``prefix-count'' 
  oracle $\calO_{\alpha_1, \dots, \alpha_C}$.
  For any string $p \in \zo^*$, the oracle $\calO_{\alpha_1, \dots, \alpha_C}(p)$
  returns the number of strings in 
  $(\alpha_1, \dots, \alpha_C)$ that begin with prefix $p$.

\paragraph{Output:} The set of all $t$-heavy hitters in $(\alpha_1, \dots, \alpha_C)$.

\paragraph{Algorithm.}
\begin{itemize}[leftmargin=10pt]
  \item Let $\heavy_0 \gets \{\epsilon\}$
        (a set containing the empty string).
  \item Let $\weight_\epsilon \gets C$.
  \item For each prefix length $\ell = 1, \dots, n$:
    \begin{itemize}[leftmargin=10pt]
          \item Let $\heavy_\ell \gets \emptyset$.
          \item For each prefix $p \in \heavy_{\ell-1}$:
                \begin{itemize}[leftmargin=10pt]
                  \item $\weight_{p\|\texttt{0}} \gets \calO_{\alpha_1, \dots, \alpha_C}(p \| \texttt{0})$, and 
                  \item $\weight_{p\|\texttt{1}} \gets \weight_p - \weight_{p\|\texttt{0}} \in \Z$.
                \end{itemize}
            Then:
            \begin{itemize}[leftmargin=10pt]
              \item If $\weight_{p \| \texttt{0}} \geq t$,
                    add $p \| \texttt{0}$ to $\heavy_\ell$.
              \item If $\weight_{p \| \texttt{1}} \geq t$,
                    add $p \| \texttt{1}$ to $\heavy_\ell$.
          \end{itemize}
        \end{itemize}
  \item Return $\heavy_n$. 
\end{itemize}

}
\end{framed}

\end{figure}

\itpara{Efficiency.}
The clients have $C$ strings total.
Then, for for any string length $\ell \in \{0,\dots, n \}$, there
are at most $C/t$ heavy strings of length $\ell$.
At each level $\ell$, the algorithm of \cref{proto:prefix} makes
at most one oracle query for each heavy string.
The algorithm thus makes at most $n \cdot C/t$ 
prefix-count-oracle queries total.
If we are looking strings that more than
a constant fraction of all clients hold (e.g., $t = 0.001C$),
then the number of queries will be independent of the number of clients.

\itpara{Security and leakage.}
While searching for the heavy hitters, the servers
will learn more information than just the $t$-heavy hitters themselves.
In particular, the servers will learn:
\begin{enumerate}
  \item[(a)] the set of all heavy strings and
  \item[(b)] for every heavy string $p$, the number of strings
             in $(\alpha_1, \dots, \alpha_n)$ that begin with $p$.
\end{enumerate}
As we discuss in \cref{sec:dp}, it is possible to apply ideas
from differential privacy to limit the damage that either type of 
the leakage can cause.

\subsection{Implementing private prefix-count queries\\ via incremental DPFs}
\label{sec:tree:alldpf}

Given the techniques of \cref{sec:tree:prefix}, we now just need to 
explain how the servers can compute the answers to prefix-count queries
over the set of client-held strings without learning anything else
about the clients' input strings.

We do this using \emph{incremental distributed point functions},
a new cryptographic primitives that builds on standard
distributed point functions (DPFs, introduced in \cref{sec:bg:private}).
Using standard DPFs for our application would also work, but would be more expensive,
both asymptotically and concretely.
If each client holds an $n$-bit string, with
plain DPFs, the client computation and
communication costs would grow as $n^2$.
With incremental DPFs, this cost falls to linear in $n$.
For our applications, $n\approx 256$, so this factor-of-$n$
performance improvement is substantial.
We get similar performance improvements on the server side.

We first define incremental DPFs, then use them to solve the heavy-hitters
problem, and finally explain how to construct them.

\begin{figure}
  \centering
  \usetikzlibrary{patterns}
\usetikzlibrary{backgrounds,intersections,calc}

\usetikzlibrary{calc}

\newcommand{\lrc}[3]{\if\side L#1\fi \if\side C#2\fi \if\side R#3\fi }

\definecolor{few-gray-bright}{HTML}{010202}
\definecolor{few-red-bright}{HTML}{EE2E2F}
\definecolor{few-green-bright}{HTML}{008C48}
\definecolor{few-blue-bright}{HTML}{185AA9}
\definecolor{few-orange-bright}{HTML}{F47D23}
\definecolor{few-purple-bright}{HTML}{662C91}
\definecolor{few-brown-bright}{HTML}{A21D21}
\definecolor{few-pink-bright}{HTML}{B43894}
\definecolor{few-gray}{HTML}{737373}
\definecolor{few-red}{HTML}{F15A60}
\definecolor{few-green}{HTML}{7AC36A}
\definecolor{few-blue}{HTML}{5A9BD4}
\definecolor{few-orange}{HTML}{FAA75B}
\definecolor{few-purple}{HTML}{9E67AB}
\definecolor{few-brown}{HTML}{CE7058}
\definecolor{few-pink}{HTML}{D77FB4}
\definecolor{few-gray-light}{HTML}{CCCCCC}
\definecolor{few-red-light}{HTML}{F2AFAD}
\definecolor{few-green-light}{HTML}{D9E4AA}
\definecolor{few-blue-light}{HTML}{B8D2EC}
\definecolor{few-orange-light}{HTML}{F3D1B0}
\definecolor{few-purple-light}{HTML}{D5B2D4}
\definecolor{few-brown-light}{HTML}{DDB9A9}
\definecolor{few-pink-light}{HTML}{EBC0DA}

\colorlet{cryptcolor}{few-green-bright}
\colorlet{cryptcolorp}{few-green-bright}

\tikzset{
  onpath/.style={fill=red!20,draw=black,text=black},
  offpath/.style={fill=black!15, text=black!70, draw=black!30
  },
  toptop/.style={draw=none,rounded corners,
    minimum width=3ex, minimum height=3ex,
    fill=none
    },
  crypt/.style={draw=cryptcolor,thick,rounded corners,
    minimum width=3ex, minimum height=3.5ex,
    fill=white},
  link/.style={draw,thick,->},
  key/.style={at=(#1.south east),
    draw=cryptcolor,fill=cryptcolor,text=white,
    anchor=center, rounded corners=0.5ex, node font=\tiny,
    inner sep=0pt, minimum width=0ex,
    text height=2.25ex, text depth=1.5ex}
}

\newcommand{\peven}[1]{\ifodd #1
0
\else 
0
\fi
}

\begin{tikzpicture}[scale=0.66, transform shape]
\foreach \side in {L,C,R} {
\let\pr\relax
\let\ph\relax

\if\side C
\tikzset{xshift=4.5cm}
\fi

\if\side R
\tikzset{xshift=9cm}
\fi

\draw[rounded corners,fill=yellow!10] (-1.2, -4.5) rectangle (3.2, -.85) {};
\newcommand{\topstyle}{toptop}

\node[\topstyle] (n) at (1,-1) {};
\node[crypt,\lrc{}{}{offpath}] (n0) at (0,-2) {\lrc{$v_0$}{$v'_0$}{$0$}};
\node[crypt,\lrc{}{}{onpath}] (n1) at (2,-2) {\lrc{$v_1$}{$v'_1$}{$\beta_1$}};

\foreach \d/\v in {1/00,2/01,3/10,4/11} {
  \node[crypt,\lrc{}{}{offpath},\if\d 3\lrc{}{}{onpath}\fi] at (\d-1.5,-3) (data\d) {\lrc{$v_{\v}$}{$v'_{\v}$}{
    \if\d 3
  $\beta_2$\else 0\fi}};
}

\lrc{
  \newcommand{\loopover}{1/000,2/001,3/010,4/011,5/100,6/101,7/110,8/111}
}{
  \newcommand{\loopover}{1/000,2/001,3/010,4/011,5/100,6/101,7/110,8/111}
}{
  \newcommand{\loopover}{1/000,2/001,3/010,4/011,5/100,7/110,8/111,6/101}
}
  \foreach \d/\v in \loopover
  {
  \node[crypt,\if \d 6\lrc{}{}{onpath}\fi,\lrc{}{}{\if \d 6\else offpath\fi}] at (0.25+\d/2-1.5,-4 - \peven{\d})
(leaf\d) 
  {\lrc{\tiny $v_{\v}$}{\tiny $v'_{\v}$}{
  \if\d 6 $\beta_3$ \else 0\fi}};}
\foreach \d in {0,1} {
  \draw[link] (n) -- (n\d);
}
\draw[link] (n0) -- (data1);
\draw[link] (n0) -- (data2);
\draw[link] (n1) -- (data3);
\draw[link] (n1) -- (data4);

\draw[link] (data1)-- (leaf1);
\draw[link] (data1)-- (leaf2);
\draw[link] (data2)-- (leaf3);
\draw[link] (data2)-- (leaf4);
\draw[link] (data3)-- (leaf5);
\draw[link] (data3)-- (leaf6);
\draw[link] (data4)-- (leaf7);
\draw[link] (data4)-- (leaf8);

\node at (1,-5) {\lrc{$\Eval(k_0, \cdot)$}{$\Eval(k_1, \cdot)$}{Sum of $\Eval$ outputs}};
}
\node[fill=white] at (3.25,-2.5) {\Huge $+$};
\node[fill=white] at (7.75,-2.5) {\Huge $=$};

\end{tikzpicture}
   \vspace{-12pt}
  \caption{Incremental DPFs give concise secret-sharing of the values on the
  nodes of a tree, such that the tree contains a single non-zero path.
  In this example, the depth $n=3$, the special point $\alpha = \texttt{101}$, 
  the values on the path are $\beta_1 \in \G_1, \beta_2 \in \G_2, \beta_3 \in \G_3$ for some 
  finite groups $\G_1$, $\G_2$, and $\G_3$, and the keys are generated as
  $\Gen(\alpha, \beta_1, \beta_2, \beta_3) \to (k_0, k_1)$.}
  \label{fig:allprefix}
\end{figure}

\paragraph{New tool: Incremental DPF.}
A standard distributed point function gives a way to succinctly secret
share a \emph{vector} of dimension $2^n$ that is non-zero at a single point.
By analogy, we can think of an incremental DPF as a secret-shared representation of 
the values on the nodes of a binary \emph{tree} with $2^n$ leaves,
where there is a single non-zero path in the tree whose nodes have non-zero values
(\cref{fig:allprefix}).

More precisely, an incremental DPF scheme, 
parameterized by finite groups~$\G_{1}, \dots, \G_{n}$, 
consists of two routines: 
\begin{itemize}
\item $\Gen(\alpha, \beta_1, \dots, \beta_n) \to (k_0, k_1)$.
      Given a string $\alpha \in \zon$ 
      and values $\beta_1 \in \G_{1}, \dots, \beta_n \in \G_{n}$, output two keys.

      We can think of the incremental DPF keys as representing secret shares
      of the values on the nodes of a tree with $2^n$ leaves and a single non-zero path.
      Using this view, $\alpha \in \zon$ is the index of the leaf at the end of
      the non-zero path.
      The values $\beta_1, \dots, \beta_n$ specify the values that the nodes along the 
      non-zero path take.
      (For simplicity, we do not assign a value to the root node of the tree.
      This is without loss of generality.)

\item $\Eval(k, x) \to \cup_{\ell=1}^n \G_{\ell}$.
      Given an incremental DPF key $k$ and string $x \in \bigcup_{\ell=1}^n \zo^\ell$,
      output a secret-shared value. 

      If we take the view of incremental DPF keys as shares of the 
      values of the nodes on a binary tree, 
      $\Eval(k, x)$ outputs a secret sharing of the value on the $x$th node
      of the tree, where we associate each node in the tree with a bitstring
      in $\bigcup_{\ell=1}^n \zo^\ell$ in the natural way.
\end{itemize}
The incremental DPF correctness property states that, for 
all strings $\alpha \in \zon$,
output values $\beta_1 \in \G_{1}, \dots, \beta_n \in \G_{n}$,
keys $(k_0,k_1) \gets \Gen(\alpha, \beta)$, and 
values $x \in \bigcup_{\ell=1}^n \zo^\ell$, it holds that
\[\Eval(k_0,\, x) + \Eval(k_1,\, x) = \begin{cases}
  \beta_\ell &\text{if $\abs{x} = \ell$ and}\\[-5pt]
  &\text{$x$ is a prefix of $\alpha$}\\
0 &\text{otherwise}
\end{cases},\]
where $\abs{x} = \ell$ and 
the addition is computed in the finite group $\G_{\ell}$.
Informally, the DPF security property states that
an adversary that learns either $k_0$ or $k_1$
(but not both) learns no information about the
special point $\alpha$ or the values $\beta_1, \dots, \beta_n$.

\medskip

We can use standard DPFs in a black-box way to build incremental DPFs:
we secret share the values at each of the $n$ levels of the tree
using a single pair of DPF keys.
With state-of-the-art DPFs, the resulting construction 
has key size and evaluation time proportional to $n^2$,
if $\alpha \in \zon$.

In contrast, we give a direct construction of incremental DPFs
from pseudorandom generators (PRGs)
that has essentially optimal key size and evaluation time. 
More specifically, each incremental DPF key has bitlength
$O(\lambda \cdot n) + \sum_{\ell=1}^n \log_2 \abs{\G_{\ell}}$,
when instantiated with a length-doubling PRG that uses $\lambda$-bit
keys and achieves $\Omega(\lambda)$-bit security.
We describe our construction in \cref{sec:inc}.

\paragraph{Using incremental DPFs to implement heavy hitters.}
We now explain how to build a system for computing $t$-private
heavy hitters using incremental DPFs (\cref{proto:heavy}).

At a high level, each client $i$ produces a pair of incremental DPF
keys that represent the secret sharing of a prefix tree that is zero
everywhere, but whose nodes have value $1$ on the path down 
to client $i$'s input string $\alpha_i$.

Given incremental DPF keys from all $C$ clients, the two servers
can compute the answers to prefix-count queries by publishing a single message each.
To compute the number of client strings that start with a prefix
$p \in \zo^*$, each server evaluates all of the
clients' incremental DPF keys on the prefix $p$ and
outputs the sum of these evaluations.

Using this technique, the servers can run the protocol of
\cref{proto:prefix} to find all of the $t$-heavy hitters.

\begin{figure}[t]
\begin{framed}
{\small
\Protocol{Private $t$-heavy hitters (semi-honest secure version)}\label[proto]{proto:heavy}
Our full protocol uses sketching to achieve security against
malicious clients (\cref{sec:mal}). We elide the sketching step here
for clarity.
There are two servers and $C$ clients.
Each client $i$, for $i \in [C]$, holds a string $\alpha_i \in \zon$.
The servers want to learn the set of all $t$-heavy hitters in 
$(\alpha_1, \dots, \alpha_C)$.
The incremental DPF works over the additive group of a finite field $\F$ where 
$\abs{\F} > C$.

The protocol is as follows:
\begin{enumerate}
\item Each client $i \in \{1, \dots, C\}$, on input string $\alpha_i \in \zon$,
      sets $\beta_1 = \cdots = \beta_n = 1 \in \F$ and prepares a pair of 
      incremental DPF keys as 
      \[ (k^{(i)}_{0}, k^{(i)}_1) \gets \Gen(\alpha_i, \beta_1, \dots, \beta_n).\]
      The client sends key $k^{(i)}_0$ to Server~$0$ and key $k^{(i)}_1$ to Server~$1$.
      After sending this single message to the servers, Client~$i$ can go offline.
    \item \label{step:oracle} The servers jointly run \cref{proto:prefix}.
      Whenever that algorithm makes a prefix-count 
      oracle query on a prefix string $p \in \zo^*$,
      each server $b\in \zo$ computes and publishes the value
      \[ \val_{p,b} \gets \sum_{i=1}^C \Eval(k^{(i)}_{b}, p) \quad \in \F.\]
      Both servers recover the answer to the prefix-count oracle query as
      \[ \val_{p} \gets \val_{p,0}+ \val_{p,1} \quad \in \F.\]
\item The servers output whatever the algorithm 
      of \cref{proto:prefix} outputs.
\end{enumerate}
}
\end{framed}
\end{figure}

\itpara{Efficiency.}
The client-to-server communication consists of a single
incremental DPF key.
The server-to-server communication requires a number of
field elements proportional to the number of prefix-count oracle
queries that the servers make.
As we argued in \cref{sec:tree:prefix}, this is at most
$n \cdot C/t$.

\itpara{Semi-honest security.}
If all parties (clients and servers) follow the protocol,
then a semi-honest adversary controlling one of the two servers
learns no more about the client strings $(\alpha_1, \dots, \alpha_C)$
that what the servers learn from the heavy-hitters 
algorithm of \cref{proto:prefix}.
\cref{sec:dp} discusses how to use ideas from differential privacy to
ameliorate the effects of this leakage.
In principle, it also would be possible for the servers to use
a constant-sized secure two-party
computation~\cite{Yao} to reduce the leakage to
a single bit per prefix-count oracle query.
Since this approach is substantially more complicated to 
implement, and since our protocol's leakage is already quite modest,
we do not discuss this direction further.

\medskip

In practice, clients and servers may deviate from the protocol,
which we discuss here:

\paragraph{Protection against malicious clients.}
As in \cref{sec:mal}, malicious clients can submit malformed incremental DPF
keys with the goal of corrupting or over-influencing the output of the protocol.
We can protect against malicious clients here using our tools from \cref{sec:mal}.

In particular, the servers will run the protocol of \cref{proto:heavy}, instantiated with the
$t$-heavy-hitters algorithm of \cref{proto:prefix}.
In this protocol, for each prefix length $\ell = 1, \dots, n$, 
the servers assemble a set---call
it $S_\ell$---of candidate heavy prefixes of length $\ell$.
The servers will then evaluate all of the clients'
incremental DPF keys at these points.

If the client is honest, the incremental DPF keys evaluated at the points in 
$S_\ell$ will be shares of a vector that is zero everywhere with a one at 
at most a single position.
Specifically, for prefix length $\ell$, client $i$'s incremental DPF keys should
evaluate to shares of the value ``$1$'' on the $\ell$-bit prefix of client $i$'s
string $\alpha_i$. The keys should evaluate to zero everywhere else.

So now the servers have the same task as in \cref{sec:mal}: the servers
hold secret shares of a client-provided vector and 
the servers want to check that this vector is zero everywhere except that it
is ``$1$'' at at most a single coordinate.
Thus, to prevent misbehavior my malicious clients, at each level $\ell \in [n]$
of the tree, the servers can use our malicious-secure sketching schemes
from \cref{sec:mal} to check that this property holds.
At each level of the tree, for each client, 
the servers perform one round of malicious-secure sketching.

We use the malicious-secure sketching approach of \cref{sec:mal:sketch}, which 
requires the client to encode its data using a redundant randomized encoding.

\paragraph{Full security: Protection against malicious servers.}
Our final task is to analyze the security of the
protocol of \cref{proto:heavy} against actively
malicious behavior by one the two participating
servers.

A malicious server's only strategy to learn extra information
in \cref{proto:heavy} is to manipulate answers to the prefix-count
oracle queries using an ``additive attack.''
For example, in Step~\ref{step:oracle} of the protocol, in processing
the answer to a prefix-oracle query $p$, Server $0$ is supposed
to publish $\val_{p,0} = \sum_{i=1}^C \Eval(k^{(i)}_{0}, p)$.
If the server is malicious, it could instead publish the
value $\val_{p,0} + \Delta$, for some non-zero shift $\Delta \in \F$. 

We capture the power of this attack in our formal security definitions 
\abbr{(the full version of this work~\cite{full})}{(\cref{app:secdefs})}, which quantify the 
information that the adversary can learn from such additive attacks.
Intuitively: the adversary can essentially control which strings
are heavy hitters (and can thus learn how many honest clients hold
strings in a small set), but the adversary can do not much worse than this.
As we discuss in \cref{sec:dp}, it is possible to further limit the 
power of this leakage using differential privacy.

\paragraph{Extension: Longer strings.}
The techniques outlined so far allow for the private computation
of $t$-heavy hitters over $n$-bit strings in which each client 
sends each server an all-prefix DPF key with domain size $n$.
Each key is roughly $\lambda n \log_2 C$ bits in length, where 
$C$ is the number of participating clients and
$\lambda \approx 128$ is 
the size of a PRG seed.

In some applications, the servers might want to compute the most popular
values over relatively long strings.
For example, an operating-system vendor might want to learn the set of
popular software binaries running on clients' machines that touch certain
sensitive system files.
In this application, client $i$'s string $\alpha_i \in \zon$ is an x86 program,
which could be megabytes long. So for this application, $n \approx 2^{20}$.

When $n$ is much bigger than $\lambda$, 
we can use hashing to reduce the client-to-server
communication from $\approx \lambda n \log_2 C$ 
bits down to $\approx \lambda^2 \log_2 C + n$ bits and the round complexity 
from $\approx n$ to $\approx \lambda$.
We describe this extension in \abbr{the full version of this work~\cite{full}}{\cref{app:hashing}}.

\section{Constructing Incremental DPFs}\label{sec:incremental}\label{sec:inc}

A straightforward way to construct an incremental DPF would be to generate $n$ independent distributed point function (DPF) keys, one for each prefix length, and to evaluate $x \in \{0,1\}^\ell$ using the $\ell$-th key. Given the most efficient DPF solution \cite{BGI16-FSS}, this would yield overall 
key size and computation for all-prefix evaluation (in units of PRG invocations) both {\em quadratic} in the input bit length $n$.
In contrast, our goal is to construct a more efficient scheme for all-prefix DPF in which all these measures are linear in $n$. 
We achieve precisely this goal, leveraging the specific structure of the DPF construction of~\cite{BGI16-FSS}.

We give the formal syntax and definitions in \abbr{the full version of this work~\cite{full}}{\cref{app:dpfdef}}.
(See \cref{sec:tree:alldpf} for informal definitions.)
In the remainder of this section, 
we sketch our construction of incremental DPF.

\begin{table}
  \centering
\begin{tabular}{rrrrr} 
  & \multicolumn{2}{c}{\bf Key size} & \multicolumn{2}{c}{\bf AES operations}\\ 
  \cmidrule{2-3}
  \cmidrule{4-5}
          & Any $n$ & $n = 256$ & Any $n$ & $n=256$\\ \midrule
  DPF~\cite{BGI16-FSS} & $\approx \frac{n^2\lambda}{2}$ & 543 KB& $\approx \frac{n^2}{2}$ & 32,640 ops.\\
This work& $\approx n(\lambda+m)$ & 6.2 KB & $\approx 2n$ & 513 ops.\\ \bottomrule
\end{tabular}
  \caption{A comparison of the key size and evaluation time of two alternatives for constructing incremental DPF: using state-of-the-art DPF as a black box and the incremental DPF construction in this paper. 
In all entries of the table the input length is $n$, the PRG seed length is $\lambda=127$, the group size in intermediate levels of the tree is $2^m$, $m=62$, and the group size in  the leaves is $2^{2\lambda}$, which suffices for the extractable DPF feature. For asymptotic expressions we assume $m \le \lambda$. The exact key size in the DPF-based construction is $\frac{n(n+1)(\lambda+2)}{2} + n(\lambda+m)+2\lambda$ and in the direct incremental DPF construction the key size is $n(\lambda+m+2)+4\lambda-m$.}
\end{table}

\paragraph{Construction.} 
We construct an efficient incremental DPF scheme, whose key size and
generation/evaluation computation costs in
particular grow {\em linearly} with the input bit
length~$n$.

In the (standard) DPF construction of~\cite{BGI16-FSS}, the evaluation of a shared point function $f_{\alpha,\beta}(x): \zo^n \to \G_n$ traverses a path defined by the binary representation of $x$. The procedure generates a pseudo-random value for each node of the path and an element of the output group $\G_n$ at the termination of the path. The two matching DPF keys are set up so that the pseudo-random  value generated by the first key is sampled independently of the value generated by the other key, for every prefix of $x$ which is also a prefix of $\alpha$. However, when the paths to $x$ and $\alpha$ diverge, the evaluation procedure programs the two pseudo-random values to be {\em equal}, by using extra information encoded in the keys, which we refer to as Correction Words (CW). The evaluation procedure on two identical pseudo-random values generates identical values along the path to $x$, and the same group value for the output, ensuring that the output is $0$ if $x \neq \alpha$. However, if $x=\alpha$ then the two independent pseudo-random values, which are known at key generation time, can be corrected to share the desired output  $\beta$.  

We extend the DPF construction of Boyle et al.~\cite{BGI16-FSS} to further support prefix outputs with small overhead.
The main observation is that the intermediate pseudo-random values generated at each level of DPF evaluation satisfy the same above-described property necessary for the final output level: namely, also for a {\em prefix} $(x_1,\dots,x_\ell) \neq (\alpha_1,\dots,\alpha_\ell)$ the intermediate evaluation generates identical pseudo-random values and for $(x_1,\dots,x_\ell)=(\alpha_1,\dots,\alpha_\ell)$ it generates independent pseudo-random values. These pseudo-random values cannot be used directly to share desired intermediate outputs, as this would compromise their pseudo-randomness required for security of the remaining DPF scheme (roughly, using them twice as a one-time pad). Instead, we introduce an extra intermediate step at each level $\ell$, which first expands the intermediate pseudo-random value $\tilde s^{(\ell)}$ to two pseudo-random values: a new seed $s^{(\ell)}$ which will take the place of $\tilde s^{(\ell)}$ in the next steps of the DPF construction, and an element of the $\ell$th level output group $\G_\ell$ which will be used to generate shares of the desired $\ell$th output $\beta_\ell \in \G_\ell$. 

Ultimately, the new procedure introduces an extra PRG evaluation and group operation per level $\ell$, as well as an additional element $W^{(\ell)}_{CW}$ of the $\ell$th level group $\G_i$ within the key, to provide the desired pseudo-random to target output correction.

\abbr{
The full pseudocode of our incremental DPF construction appears in the full version of this work~\cite{full}.
It yields the following result:
}{
We proceed with a description of an optimized construction of an Incremental DPF in Figure~\ref{fig:IDPF}.
The generation of the new correction word values $W_{CW}$ is performed in lines~\ref{step:Gen-W1},~\ref{step:Gen-W2} of $\Gen$, and their usage within evaluation is in lines~\ref{step:Eval-W1},~\ref{step:Eval-W2} of $\Eval\Next$. 

The powers of $(-1)$ in $\Gen$ line~\ref{step:Gen-W1} and $\Eval\Next$ line~\ref{step:Eval-W2} are to address arbitrary output group structure $\G_\ell$, replacing xor with addition (of inverses) within the group. Here, party $b=1$ will always output the negation of his computed share, so that once again identical pseudo-random shares will yield shares of the identity 0 within $\G_\ell$; the correction word $W_{CW}$ is negated as necessary depending on whether party $b=0$ or $1$ is the one to incorporate the correction, as indicated by $t_1^{(\ell-1)}$.
}

\begin{proposition}[Incremental DPF]  \label{prop:IDPF}
  The incremental DPF scheme described in \abbr{the full version of this work~\cite{full}}{\cref{fig:IDPF}}
is a secure Incremental DPF with the following complexities for $(\alpha,(\G_1,\beta_1),\dots,(\G_n,\beta_n))$: 
  \begin{itemize}
  \item Key size: $\lambda + (\lambda+2)n + \sum_{j \in [n]} \lceil \log |\G_j| \rceil$ bits.
  
  \item Computation: Let ${\sf cost}(\ell) := 1 + \lceil \log (|\G_\ell |) / \lambda \rceil$. Units given in evaluation of a PRG $G: \zo^\lambda \to \zo^{2\lambda+2}$:
\begin{itemize}
    \item $\Gen$: $2 \sum_{\ell \in [n]} {\sf cost}(\ell)$
\item $\Eval(x)$: $\sum_{\ell \in [\abs{x}]} {\sf cost}(\ell)$
    \end{itemize}
  \end{itemize}
\end{proposition}

We prove \cref{prop:IDPF} in \abbr{the full version of this work~\cite{full}}{\cref{app:idpf}}.

In \abbr{the full version of this work~\cite{full}}{\cref{app:opts}}, we describe a number of low-level optimizations
that we have implemented to make our incremental DPF construction more
efficient, especially when using AES hardware instructions to implement
the PRG.

 \section{Providing differential privacy}
\label{sec:dp}

\abbr{}{
In many settings the set of heavy hitters itself
can leak sensitive information
about users' private inputs.
For example, say that the servers run our
heavy-hitters protocol once on a set of
client-provided URLs $(x_1, \dots, x_C)$ and the
protocol output indicates that \texttt{nytimes.com} is a heavy hitter.
Then, one of the clients goes offline.
The servers run our heavy-hitters protocol a second time
on the smaller set of $C-1$ URLs $(x_1, \dots, x_{C-1})$
and the protocol output indicates that \texttt{nytimes.com} is
\emph{not} a heavy hitter.
In this case, anyone who observes the set of
clients who participated in each protocol run
along with the public output of the protocol can
infer with certainty that the URL $x_C$ of client
$C$ was \texttt{nytimes.com}.
So even though ``nothing more'' than the set of
heavy hitters leaks, 
this information itself can be sensitive.}

To bound the amount of information that an adversary can infer
from the system's output, we can ensure that the system's 
output satisfies $\epsilon$-\emph{differential privacy}~\cite{DMNS06,DR14}.
This is possible with a simple tweak to
our heavy-hitters protocol (\cref{proto:heavy}), which we
describe in \abbr{the full version of this work~\cite{full}}{\cref{app:dpx}}.

 \section{Implementation and evaluation}
\label{sec:eval}

\begin{figure}
  \centering
  \includegraphics{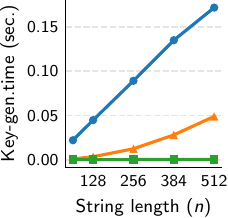}~~~~\includegraphics{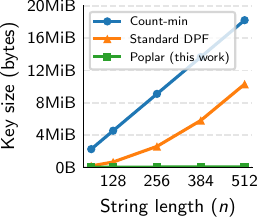}
  \caption{Client-side costs in time (left) and communication (right) of \name,
          compared with a baseline scheme using standard DPFs and a prior approach
          using count-min sketches~\cite{melis2016}.
          Costs for \name are relatively small and grow linearly in the 
          length $n$ of each client's string.}
  \label{fig:clienttime}
\end{figure}

The \name system is an end-to-end implementation of the 
private heavy-hitters scheme described in this paper.
Our implementation is roughly 3,500 lines of Rust code
(compiled with 1.46.0-nightly), including tests.
The code is online at \url{https://github.com/henrycg/heavyhitters}.

\Name's sketching scheme uses a 62-bit finite field in the middle of
the ``tree'' (\cref{fig:allprefix}) and 
a 255-bit field at the leaves.
With this configuration, our sketching schemes detect cheating clients,
except with probability $\approx 2^{-62}$, over the servers' random choices,
independent of how much computation a cheating client does.
While we expect this level of security against a cheating client
to be sufficient in practice, by running
the sketching scheme twice---at most doubling the communication and computation---we
can achieve nearly $128$-bit security.
Using the larger field at the leaves ensures that our DPF construction satisfies
the extractability property (\cref{lem:esimple}) against cheating clients that
run in time at most $\approx 2^{128}$.

\itpara{Client costs.}
\cref{fig:clienttime} shows the client costs for 
three different private heavy-hitters schemes. 
Our client experiments run on an Intel i7-1068NG7 CPU at 2.3 GHz.
The first is \name's tree-based
scheme (\cref{sec:tree}), based on our new incremental DPFs (\cref{sec:incremental}). 
The second uses our tree-based scheme, but with standard DPFs~\cite{BGI16-FSS}.
The third uses private aggregation of count-min sketches~\cite{countmin,melis2016}
to compute approximate heavy hitters.
For the count-min sketches, we set the approximation error $\epsilon=1/128$ and
failure probability $\delta=2^{-60}$. (To reduce communication in this third scheme,
we use DPFs here as well.)

Our incremental DPF keys have size linear in the length of the clients' strings,
with a small constant.
In contrast, using standard DPFs requires one linear-sized key for each layer of
the prefix tree (\cref{sec:tree}), which yields a quadratic cost.
The count-min-sketch based private aggregation scheme also has a linear client-side cost,
but the large size of each sketch makes the constant substantially worse.

\itpara{Server communication.}
\cref{fig:comm} shows the total communication cost
per server per client of running \name's end-to-end
heavy-hitters protocol.
In this experiment, clients sample their strings from a
Zipf distribution with
parameter 1.03 and support 10,000.
This type of ``power-law'' distribution arises naturally in
network settings~\cite{KL01} and we choose the
parameter conservatively (i.e., the distribution
is closer to uniform than we would expect in nature), 
which likely gives an underestimate of \name's performance.
In this experiment, servers search for
strings that more than $0.1\%$ of clients hold.
In \name, the total communication per client is 
tens of kilobytes.
\cref{fig:comm} also estimates the dollar cost of computing
private heavy hitters using the baseline scheme (based on standard DPFs)
and \name, as the number of clients varies.
\Name is roughly two orders of magnitude less expensive.
\begin{figure}
  \centering
  \includegraphics{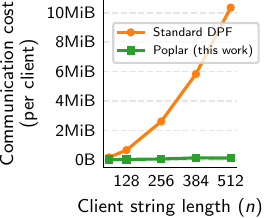}~~~~\includegraphics{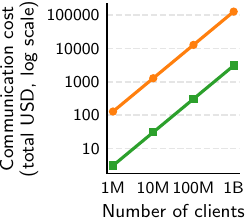}
  \caption{Total server communication cost (send + receive) per client of \name.
          At left: we simulate the server-side communication cost \name
          and compare to a scheme based on standard DPFs.
          At right: we compare the US-dollar cost of the schemes
          used with 256-bit strings searching for the top-900 heavy hitters,
          using Amazon EC2's current (Dec.\@ 2020) data-transfer price of USD 0.05/GB.}
  \label{fig:comm}
\end{figure}

\itpara{End-to-end performance.}
Finally, we ran an end-to-end performance test of \name
over the Internet.
We use one c4.8xlarge server (32 virtual cores) 
in Amazon's us-east-1 region (N.~Virginia) 
and one in the us-west-1 region (N.~California).
The round-trip latency between the two data centers was 61.8ms.
We measure the running time from the moment after the servers
collect the last incremental DPF keys from the clients until the
servers produce their output.
Each client holds a 256-bit string, which is enough to represent a 
42-character domain name (uncompressed).
\cref{tab:ec2table} shows the results of this experiment.
For 400,000 clients, the total running time is around 53 minutes.

\Name is almost completely parallelizable.
In \cref{fig:simtime}, we give estimates for the protocol-execution time, as
a function of the number of clients and the number of physical machines used
to implement each of the system's two logical servers.
When deployed with 20 machines per logical server, we estimate that
\name could process ten million client requests in just over one hour.

\begin{table}[t]
  \centering
  
\centering
\begin{tabular}{rrrrrr}
  & \multicolumn{3}{c}{\textbf{Running time (sec.)}} & & \\ \cmidrule{2-4}
  \textbf{Clients} & DPF & Sketching & Total & Clients/Sec.\\ \midrule
100k& 107.3 & 704.5 & 828.1  & 120.8 \\
200k& 211.0 & 1,404.1  & 1,633.5 & 122.4 \\
400k& 433.5  & 2,771.4 & 3,226.0 & 124.0 \\ \midrule
\end{tabular}
   \caption{End-to-end cost of \name, when used to collect
  $n = 256$-bit strings. Each client's string is sampled from a Zipf distribution 
  with parameter $1.03$ and support $10,000$.
  We use one c4.8xlarge server (32 virtual cores) to implement each of the two logical servers.
  One server is in Amazon's N.~California data center and the other is in N.~Virginia.
  The servers set the heavy-hitters threshold to collect all strings that more than
  $0.1\%$ of clients hold.}
  \label{tab:ec2table}
\end{table}
\begin{figure}
  \centering
  \begin{minipage}{0.66\columnwidth} 
  \includegraphics{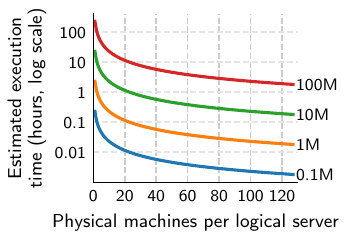}
  \end{minipage}~\begin{minipage}{0.32\columnwidth} 
  \caption{Estimated execution time of the \name system.
    Each line represents a number of clients.
    The workload is fully parallelizable, so we
    neglect sharding costs.
    System parameters are as in \cref{tab:ec2table}.}
  \label{fig:simtime}
  \end{minipage}
\end{figure}

 \section{Conclusions}

We have described \name, a system that allows two non-colluding servers
to compute the most popular strings among a large set of client-held strings while
preserving client privacy.
Along the way, we introduced several lightweight cryptographic tools: a 
protocol for checking that a secret-shared vector is a unit vector, 
an extractable variant of distributed point functions that defends against badly formed keys,
and a generalization of distributed point functions for secret-sharing weights on binary trees.

There are a number of potential extensions to this work. 
For instance, instead of finding heavy hitters, the servers
might like to find heavy \emph{clusters}---strings
that are close to many of the client-held strings,
under some distance metric.
Perhaps each client holds a GPS
coordinate pair and the servers want to learn the
popular neighborhoods.

\paragraph{Acknowledgments.}
We thank Eric Rescorla for suggesting this problem,
Saba Eskandarian for helpful comments, 
Phillipp Schoppmann and Simone Colombo for pointing out typos,
Christopher Patton for technical discussion,
Mayank Rathee for pointing out a bug in \cref{fig:ConvertFn} in 
the proceedings version of this work,
and the anonymous reviewers for their feedback and suggestions.
Dan Boneh was funded by NSF, DARPA, a grant from ONR, and the Simons Foundation.
Elette Boyle was supported by ISF grant 1861/16, AFOSR Award FA9550-17-1-0069, and
ERC Project HSS (852952).
Henry Corrigan-Gibbs was funded in part by NSF (CNS-2054869), Facebook, and Google.
Henry thanks Bryan Ford for generously hosting him at EPFL during the early stages of this project.
Niv Gilboa was supported by ISF grant 2951/20, ERC grant 876110, and a grant by
the BGU Cyber Center.
Yuval Ishai was supported by ERC Project NTSC (742754), ISF grant 2774/20,
NSF-BSF grant 2015782, and BSF grant 2018393. 

\frenchspacing
{\small
\bibliographystyle{plain}
\bibliography{refs}

\begin{thebibliography}{10}

\bibitem{Blinder}
Ittai Abraham, Benny Pinkas, and Avishay Yanai.
\newblock Blinder: {MPC} based scalable and robust anonymous committed
  broadcast., 2020.

\bibitem{ARFCR10}
Benny Applebaum, Haakon Ringberg, Michael~J Freedman, Matthew Caesar, and
  Jennifer Rexford.
\newblock Collaborative, privacy-preserving data aggregation at scale.
\newblock In {\em PETS}, pages 56--74. Springer, 2010.

\bibitem{BBGN20}
Borja Balle, James Bell, Adria Gascon, and Kobbi Nissim.
\newblock Private summation in the multi-message shuffle model.
\newblock 2020.

\bibitem{BNST17}
Raef Bassily, Kobbi Nissim, Uri Stemmer, and Abhradeep~Guha Thakurta.
\newblock Practical locally private heavy hitters.
\newblock In {\em Neural Information Processing Systems}, pages 2288--2296,
  2017.

\bibitem{BS15}
Raef Bassily and Adam Smith.
\newblock Local, private, efficient protocols for succinct histograms.
\newblock In {\em Proceedings of the forty-seventh annual ACM symposium on
  Theory of computing}, pages 127--135, 2015.

\bibitem{B91}
Donald Beaver.
\newblock Efficient multiparty protocols using circuit randomization.
\newblock In {\em CRYPTO}, pages 420--432. Springer, 1991.

\bibitem{BR93}
Mihir Bellare and Phillip Rogaway.
\newblock Random oracles are practical: A paradigm for designing efficient
  protocols.
\newblock In {\em CCS}, pages 62--73, 1993.

\bibitem{flpcp}
Dan Boneh, Elette Boyle, Henry Corrigan-Gibbs, Niv Gilboa, and Yuval Ishai.
\newblock Zero-knowledge proofs on secret-shared data via fully linear {PCPs}.
\newblock In {\em CRYPTO}, pages 67--97. Springer, 2019.

\bibitem{BGI19}
Elette Boyle, Niv Gilboa, and Yuval Ishai.
\newblock Secure computation with preprocessing via function secret sharing.
\newblock In Dennis Hofheinz and Alon Rosen, editors, {\em {TCC} 2019}, pages
  341--371.

\bibitem{BGI15}
Elette Boyle, Niv Gilboa, and Yuval Ishai.
\newblock Function secret sharing.
\newblock In {\em EUROCRYPT}, 2015.

\bibitem{BGI16-FSS}
Elette Boyle, Niv Gilboa, and Yuval Ishai.
\newblock Function secret sharing: Improvements and extensions.
\newblock In {\em CCS}, 2016.

\bibitem{BNS19}
Mark Bun, Jelani Nelson, and Uri Stemmer.
\newblock Heavy hitters and the structure of local privacy.
\newblock {\em ACM Transactions on Algorithms (TALG)}, 15(4):1--40, 2019.

\bibitem{Sepia}
Martin Burkhart, Mario Strasser, Dilip Many, and Xenofontas Dimitropoulos.
\newblock {SEPIA}: Privacy-preserving aggregation of multi-domain network
  events and statistics.
\newblock {\em USENIX Security}, 2010.

\bibitem{Canetti00}
Ran Canetti.
\newblock Security and composition of multiparty cryptographic protocols.
\newblock {\em J. Cryptology}, 13(1):143--202, 2000.

\bibitem{CSS11}
T-H~Hubert Chan, Elaine Shi, and Dawn Song.
\newblock Private and continual release of statistics.
\newblock {\em ACM Transactions on Information and System Security (TISSEC)},
  14(3):1--24, 2011.

\bibitem{C81}
David~L Chaum.
\newblock Untraceable electronic mail, return addresses, and digital
  pseudonyms.
\newblock {\em Communications of the ACM}, 24(2):84--90, 1981.

\bibitem{CKMS03}
Graham Cormode, Flip Korn, Shanmugavelayutham Muthukrishnan, and Divesh
  Srivastava.
\newblock Finding hierarchical heavy hitters in data streams.
\newblock In {\em VLDB}, 2003.

\bibitem{countmin}
Graham Cormode and S~Muthukrishnan.
\newblock An improved data stream summary: the count-min sketch and its
  applications.
\newblock {\em Journal of Algorithms}, 55(1):58--75, 2005.

\bibitem{Prio}
Henry Corrigan-Gibbs and Dan Boneh.
\newblock Prio: Private, robust, and scalable computation of aggregate
  statistics.
\newblock In {\em NSDI}, pages 259--282, 2017.

\bibitem{Riposte}
Henry Corrigan-Gibbs, Dan Boneh, and David Mazi{\`e}res.
\newblock Riposte: An anonymous messaging system handling millions of users.
\newblock In {\em IEEE Symposium on Security and Privacy}, 2015.

\bibitem{AMD08}
Ronald Cramer, Yevgeniy Dodis, Serge Fehr, Carles Padr{\'{o}}, and Daniel
  Wichs.
\newblock Detection of algebraic manipulation with applications to robust
  secret sharing and fuzzy extractors.
\newblock In {\em {EUROCRYPT}}, pages 471--488, 2008.

\bibitem{SPDZ}
Ivan Damg{\aa}rd, Valerio Pastro, Nigel~P. Smart, and Sarah Zakarias.
\newblock Multiparty computation from somewhat homomorphic encryption.
\newblock In {\em {CRYPTO}}, pages 643--662, 2012.

\bibitem{DFKZ13}
George Danezis, C{\'e}dric Fournet, Markulf Kohlweiss, and Santiago
  Zanella-B{\'e}guelin.
\newblock Smart meter aggregation via secret-sharing.
\newblock In {\em Workshop on Smart Energy Grid Security}, pages 75--80. ACM,
  2013.

\bibitem{DS17}
Jack Doerner and Abhi Shelat.
\newblock Scaling {ORAM} for secure computation.
\newblock In {\em CCS}, pages 523--535, 2017.

\bibitem{DKMMN06}
Cynthia Dwork, Krishnaram Kenthapadi, Frank McSherry, Ilya Mironov, and Moni
  Naor.
\newblock Our data, ourselves: Privacy via distributed noise generation.
\newblock In {\em EUROCRYPT}, pages 486--503. Springer, 2006.

\bibitem{DMNS06}
Cynthia Dwork, Frank McSherry, Kobbi Nissim, and Adam Smith.
\newblock Calibrating noise to sensitivity in private data analysis.
\newblock In {\em Theory of Cryptography Conference}, pages 265--284. Springer,
  2006.

\bibitem{DR14}
Cynthia Dwork, Aaron Roth, et~al.
\newblock The algorithmic foundations of differential privacy.
\newblock {\em Foundations and Trends in Theoretical Computer Science},
  9(3-4):211--407, 2014.

\bibitem{elahi2014privex}
Tariq Elahi, George Danezis, and Ian Goldberg.
\newblock {PrivEx}: Private collection of traffic statistics for anonymous
  communication networks.
\newblock In {\em CCS}, pages 1068--1079. ACM, 2014.

\bibitem{PrivEx}
Tariq Elahi, George Danezis, and Ian Goldberg.
\newblock Privex: Private collection of traffic statistics for anonymous
  communication networks.
\newblock In {\em CCS}, pages 1068--1079, 2014.

\bibitem{Rappor}
{\'U}lfar Erlingsson, Vasyl Pihur, and Aleksandra Korolova.
\newblock Rappor: Randomized aggregatable privacy-preserving ordinal response.
\newblock In {\em CCS}, pages 1054--1067, 2014.

\bibitem{Express}
Saba Eskandarian, Henry Corrigan-Gibbs, Matei Zaharia, and Dan Boneh.
\newblock Express: Lowering the cost of metadata-hiding communication with
  cryptographic privacy.
\newblock {\em arXiv preprint arXiv:1911.09215}, 2019.

\bibitem{RapporII}
Giulia Fanti, Vasyl Pihur, and {\'U}lfar Erlingsson.
\newblock Building a {Rappor} with the unknown: Privacy-preserving learning of
  associations and data dictionaries.
\newblock {\em Proceedings on Privacy Enhancing Technologies}, 2016(3):41--61,
  2016.

\bibitem{GLO15}
Sanjam Garg, Steve Lu, and Rafail Ostrovsky.
\newblock Black-box garbled {RAM}.
\newblock In {\em FOCS}, pages 210--229. IEEE, 2015.

\bibitem{GenkinIPST14}
Daniel Genkin, Yuval Ishai, Manoj Prabhakaran, Amit Sahai, and Eran Tromer.
\newblock Circuits resilient to additive attacks with applications to secure
  computation.
\newblock In David~B. Shmoys, editor, {\em Symposium on Theory of Computing,
  {STOC}}, pages 495--504. {ACM}, 2014.

\bibitem{GI14}
Niv Gilboa and Yuval Ishai.
\newblock Distributed point functions and their applications.
\newblock In {\em EUROCRYPT}, pages 640--658, 2014.

\bibitem{GMW87}
O~Goldreich, S~Micali, and A~Wigderson.
\newblock How to play any mental game.
\newblock In {\em STOC}, pages 218--229, 1987.

\bibitem{Goldbook}
Oded Goldreich.
\newblock {\em Foundations of Cryptography II: Basic Applications}.
\newblock Cambridge University Press, 2004.

\bibitem{GO96}
Oded Goldreich and Rafail Ostrovsky.
\newblock Software protection and simulation on oblivious rams.
\newblock {\em Journal of the ACM}, 43(3):431--473, 1996.

\bibitem{GKKKMRV12}
S~Dov Gordon, Jonathan Katz, Vladimir Kolesnikov, Fernando Krell, Tal Malkin,
  Mariana Raykova, and Yevgeniy Vahlis.
\newblock Secure two-party computation in sublinear (amortized) time.
\newblock In {\em Proceedings of the 2012 ACM conference on Computer and
  communications security}, pages 513--524, 2012.

\bibitem{JJ16}
Rob Jansen and Aaron Johnson.
\newblock Safely measuring {Tor}.
\newblock In {\em CCS}, pages 1553--1567, 2016.

\bibitem{JK12}
Marek Jawurek and Florian Kerschbaum.
\newblock Fault-tolerant privacy-preserving statistics.
\newblock In {\em PETS}, pages 221--238. Springer, 2012.

\bibitem{KY18}
Marcel Keller and Avishay Yanai.
\newblock Efficient maliciously secure multiparty computation for {RAM}.
\newblock In {\em EUROCRYPT}, pages 91--124. Springer, 2018.

\bibitem{KL01}
Jon Kleinberg and Steve Lawrence.
\newblock The structure of the web.
\newblock {\em Science}, 294(5548):1849--1850, 2001.

\bibitem{LO13}
Steve Lu and Rafail Ostrovsky.
\newblock Distributed oblivious {RAM} for secure two-party computation.
\newblock In {\em Theory of Cryptography Conference}, pages 377--396. Springer,
  2013.

\bibitem{melis2016}
Luca Melis, George Danezis, and Emiliano De~Cristofaro.
\newblock Efficient private statistics with succinct sketches.
\newblock In {\em NDSS}. Internet Society, February 2016.

\bibitem{hijack}
Mozilla.
\newblock Your browser is hijacked, now what?
\newblock
  \url{https://blog.mozilla.org/firefox/your-browser-is-hijacked-now-what/},
  Accessed 19 August 2020, October 2018.

\bibitem{N01}
C~Andrew Neff.
\newblock A verifiable secret shuffle and its application to e-voting.
\newblock In {\em CCS}, pages 116--125, 2001.

\bibitem{vpriv}
Raluca~Ada Popa, Hari Balakrishnan, and Andrew~J. Blumberg.
\newblock {VPriv}: Protecting privacy in location-based vehicular services.
\newblock In {\em USENIX Security}, pages 335--350, 2009.

\bibitem{qin2016heavy}
Zhan Qin, Yin Yang, Ting Yu, Issa Khalil, Xiaokui Xiao, and Kui Ren.
\newblock Heavy hitter estimation over set-valued data with local differential
  privacy.
\newblock In {\em {CCS}}, 2016.

\bibitem{Adnostic}
Vincent Toubiana, Arvind Narayanan, Dan Boneh, Helen Nissenbaum, and Solon
  Barocas.
\newblock Adnostic: Privacy preserving targeted advertising.
\newblock In {\em NDSS}, 2010.

\bibitem{Splinter}
Frank Wang, Catherine Yun, Shafi Goldwasser, Vinod Vaikuntanathan, and Matei
  Zaharia.
\newblock Splinter: Practical private queries on public data.
\newblock In {\em NSDI}, pages 299--313, 2017.

\bibitem{Dissent}
David~Isaac Wolinsky, Henry Corrigan-Gibbs, Aaron Johnson, and Bryan Ford.
\newblock Dissent in numbers: Making strong anonymity scale.
\newblock In {\em 10th OSDI}. USENIX, October 2012.

\bibitem{Yao}
Andrew Chi-Chih Yao.
\newblock How to generate and exchange secrets.
\newblock In {\em FOCS}, pages 162--167. IEEE, 1986.

\bibitem{ZKMSL20}
Wennan Zhu, Peter Kairouz, Brendan McMahan, Haicheng Sun, and Vivian~(Wei) Li.
\newblock Federated heavy hitters with differential privacy.
\newblock In {\em AISTATS}, 2020.

\end{thebibliography}
}
\nonfrenchspacing

\appendix
\abbr{}{\section{Formal security definitions}
\label{app:secdefs}

To formally specify the security properties of our protocols, we use the standard ``real vs.\ ideal'' definition paradigm for secure multiparty computation~\cite{Canetti00,Goldbook}. This involves specifying a precise ideal functionality or leakage function for each type of corruption. We will start with the simpler case of the subset-histogram protocol, and then address the heavy-hitters protocol.

For both protocols, we consider here the ``bare-bones'' version that does not add noise for differential privacy purposes. The differentially private variant, discussed in Section~\ref{sec:dp}, adds a suitable amount of server-generated noise to the functionalities described below. This extra defense mechanism may not be needed in situations where there are good statistical guarantees on the entropy of the inputs contributed by honest clients.

In the following, security refers to computational security with respect to a common security parameter $\lambda$ that is given to all parties. The security of our protocols against malicious clients is proved in the random oracle model. Here we assume that inputs of honest parties are picked independently of the oracle.

\paragraph{Functionalities for subset-histogram protocol.}
\begin{itemize}
\item \emph{Parties:} $C$ client parties, two servers. 
\item \emph{Public parameters:} String length $n$, upper bound on $C$.
\item \emph{Functionality for honest parties:}
\begin{itemize}
\item Receive input $S=\{\sigma_1,\ldots,\sigma_m\}\subseteq \{0,1\}^n$ from servers.
\item For $i \in [C]$, receive input $\alpha_i \in \zon$ from Client $i$. 
\item Deliver to  servers the output \[\textit{agg} = f_S(\alpha_1, \ldots, \alpha_C)=\langle (\sigma_1,\weight_1),\ldots,(\sigma_m,\weight_m)\rangle, \] where $\weight_j=|\{i\,:\,\alpha_i=\sigma_j\}|$.
\end{itemize}
\item \emph{Functionality for malicious clients:} Suppose clients $T\subset [C]$ are controlled by an efficient malicious adversary $A$. The influence of $A$ on the output is captured by the following functionality.
\begin{itemize}
\item Receive input $S=\{\sigma_1,\ldots,\sigma_m\}\subseteq \{0,1\}^n$ from servers.
\item For $i \in [C]\setminus T$, receive $\alpha_i \in \zon$ from Client $i$. 
\item For $i \in T$, receive $\alpha^*_i \in \zon$ and selective vote predicate $V^*_i:\{0,1\}^n\to\{0,1\}$, specified by a Boolean circuit,
from $A$. (In the random-oracle model, $V^*_i$ can invoke the oracle.) Intuitively, if $V^*_i(\sigma)=1$ for some $\sigma\in S$, the vote $\alpha^*_i$ of Client $i$ will not count.
\item Deliver to  servers $\textit{agg} = f_S(\alpha'_1, \ldots, \alpha'_C)$, where: 
\begin{enumerate}
\item $\alpha'_i=\alpha_i$ if $i\not\in T$,
\item $\alpha'_i=\alpha^*_i$ if $i\in T$ and moreover $V^*_i(\sigma)=0$ for all $\sigma\in S$,   
\item $\alpha'_i=\bot$ otherwise.
\end{enumerate}
\end{itemize}
\item \emph{Leakage for malicious server:} Suppose both Server $b$, for $b\in\{0,1\}$, and client set $T\subset [C]$ are controlled by an efficient malicious adversary $A$. The view of $A$ can be simulated given the following leakage function.
\begin{itemize}
\item Receive  $S=\{\sigma_1,\ldots,\sigma_m\}\subseteq \{0,1\}^n$ from honest Server $1-b$ and leak it to $A$.
\item Let $D\leftarrow A(S)$, where $D\subseteq [C]$,  be a subset of disqualified clients. 
\item Let $Q= [C]\setminus (T\cup D)$.
\item For $i \in Q$, receive input $\alpha_i \in \zon$ 
from Client $i$. 
\item Leak $(S,Q,\textit{agg}_Q)$ to $A$, where $\textit{agg}_Q = f_S((\alpha_i)_{i\in Q})$.
\end{itemize}

\end{itemize}

\paragraph{Functionalities for heavy-hitters protocol.}
\begin{itemize}
\item \emph{Parties:} $C$ client parties, two servers. 
\item \emph{Public parameters:} String length $n$, upper bound on $C$.
\item \emph{Functionality for honest parties:}
\begin{itemize}
\item Receive heavy-hitter threshold $t=\tau C$ from servers.
\item For $i \in [C]$, receive input $\alpha_i \in \zon$
from Client $i$. 
\item Deliver to servers the output \[\textit{agg} = f_t(\alpha_1, \ldots, \alpha_C)= \{ \alpha \,:\, |\{ i:\alpha_i=\alpha\}|\ge t \}. \]
\end{itemize}
\item \emph{Functionality for malicious clients:} Suppose clients $T\subset [C]$ are controlled by an efficient malicious adversary $A$. The influence of $A$ on the output is captured by the following functionality.
\begin{itemize}
\item Receive threshold $t$ from servers.
\item For $i \in [C]\setminus T$, receive $\alpha_i \in \zon$ from Client $i$. 
\item For $i \in T$, receive $\alpha^*_i \in \zon$ and selective vote predicate $V^*_i:P_n \to\{0,1,\bot\}$
from $A$, where $P_n$ denotes the set of binary strings of length $\le n$ and $V^*_i$ is specified by a Boolean circuit. (In the random-oracle model, $V^*_i$ can invoke the oracle.) For $p$ such that $|p|=n$, we require that if $V^*_i(p)=1$ then $p=\alpha^*_i$ (for any choice of the oracle). 
Intuitively, $V_i^*(p)$ represents the vote of Client $i$ for prefix $p$, under the restriction of casting at most one vote for an $n$-bit string. (Alternatively, with a slight loss of efficiency, we can realize a variant that restricts a client to one vote for every prefix length.)
\item Run Algorithm~\ref{proto:prefix}, with the following modifications: 
\begin{itemize}
\item Initialize a set $\hat C\leftarrow [C]$ of active clients. 
\item Before iteration $\ell$, $1\le \ell\le n$, remove from $\hat C$ every client~$i$ for which either: (1) there is $p^*\in H_{\ell-1}$ such that $V^*_i(p^*\|0)=\bot$, or (2) there are two distinct $p^*\in H_{\ell-1}$ such that $V^*_i(p^*\|0)=1$.
\item In iteration $\ell$, compute each weight $\weight_{p\|0}$ by $\weight_{p\|0}=\sum_{i\in\hat C}V'_i(p\|0)$, where $V'_i=V^*_i$ if $i\in T$, and if $i\not\in T$ then $V'_i(\beta)$ returns $1$ if $\beta$ is a prefix of $\alpha_i$ and 0 otherwise.
\end{itemize}
\item Deliver to both servers the output $H_n$ of Algorithm~\ref{proto:prefix}.
\end{itemize}
\item \emph{Leakage for malicious server:} Suppose both Server $b$, for $b\in\{0,1\}$, and client set $T\subset [C]$ are controlled by an efficient malicious adversary $A$. The view of $A$ can be simulated given the following leakage function.
\begin{itemize}
\item Receive heavy-hitter threshold $t=\tau C$ from honest Server $1-b$ and leak it to $A$.
\item For $i \in [C]\setminus T$, receive $\alpha_i \in \zon$ from Client $i$. 
\item Run Algorithm~\ref{proto:prefix}, with the following modifications: 
\begin{itemize}
\item Initialize set $\hat C\leftarrow [C]\setminus T$ of active honest clients. 
\item Before iteration $\ell$, $1\le \ell\le n$, allow $A$ to choose a set of clients to remove from $\hat C$. 
\item In iteration $\ell$, compute a tampered version $\weight^*_{p\|0}$ of the weights $\weight_{p\|0}$, defined by \[ \weight^*_{p\|0}=\sum_{i\in\hat C}V_i(p\|0)+\Delta_p,\] where $V_i(\beta)$ outputs 1 if $\beta$ is a prefix of $\alpha_i$ and 0 otherwise, and $\Delta_p$ is an integer chosen by $A$ based on all previous information it learned. Leak the value $\weight^*_{p\|0}$ to $A$.
\item  If the set $H_\ell$ computed in the end of iteration $\ell$ satisfies $|H_\ell|>1/\tau$, abort.
\end{itemize}
\end{itemize}
\end{itemize}

 }
\abbr{}{

\section{Extension: Hashing for longer strings}
\label{app:hashing}

In this section we describe a hashing-based optimization that improves the communication complexity and the round complexity of our heavy hitters protocol when $n\gg \lambda$.  Before describing our solution, we start with a simpler approach that fails to achieve security against malicious clients.

A first idea is to have all clients use a public random hash function $\Hash \colon \zon \to \zo^{2\lambda}$ to map their long $n$-bit inputs into $2\lambda$-bit strings,\footnote{Here and in the following, we choose a $2\lambda$-bit output to ensure that collisions occur with negligible probability. In some settings a shorter output size would suffice, depending on the number of clients $C$ and the tolerable error probability.} and run the heavy hitters protocol on the shorter inputs $\Hash(\alpha_i)$. An obvious problem is that this only reveals the popular hash values $\Hash(\alpha)$ instead of the popular strings $\alpha$. A natural fix is to have each client write its full string $\alpha_i$ at the corresponding IDPF leaf. More precisely, we use the payload group $\G=\F\times \F^{n'}$, where $\F$ is a prime field such that $|\F|>C$ and the strings in $\zon$ are encoded as vectors in $\F^{n'}$. Each client $i$ adds $(1,\alpha_i)$ to the payload of leaf $\Hash(\alpha_i)$. Aggregating the contributions of all clients, and assuming no hash collisions occur, leaf $\Hash(\alpha)$ contains the payload $(v,v\cdot\alpha)$ where $v$ is the number of clients with input $\alpha$. Note that $\alpha$ can be fully recovered from the payload, as required.

The above solution achieves our efficiency and security goals when all clients are semi-honest. (Efficiency results from the fact that the payload size is only an additive term in the IDPF key size, and has no influence on the round complexity of our heavy hitters protocol.) However, even just a single malicious client can easily corrupt the information about a string $\alpha$ by writing a {\em random} payload to leaf $\Hash(\alpha)$. 

To mitigate this attack, we use the following approach. We view each leaf $\Hash(\alpha)$ as the root of a depth-$\lambda$ binary subtree, and let each client $i$ write the (long) string $\alpha_i\in\zon$ to a random leaf of the tree rooted by $\Hash(\alpha_i)$. Since our heavy hitters protocol prevents a malicious client from writing to more than a single leaf, the probability of any malicious client colliding with a string written by an honest client is negligible. By keeping in each node of the extended tree a count of the number of the non-empty leaves in its subtree, the servers can traverse the subtree rooted by each heavy hitter hash value $h=\Hash(\alpha)$ until they find a leaf containing a long string $\alpha$ consistent with $h$. This approach can be enhanced by using an error-correcting code for encoding the inputs, where each client writes a random symbol of the encoding of $\alpha_i$. 

In more detail, our solution proceeds as follows. 
We use a hash function: $\Hash \colon \zon \to \zo^{2\lambda}$,
which we model as a random oracle~\cite{BR93}.

Each client $i$ runs the following steps:
\begin{itemize}
  \item Compute $h_i \gets \Hash(\alpha_i) \in \zo^{2\lambda}$.
  \item Choose a random nonce $\nu_i \getsr \zo^\lambda$ and
        set $\hat \alpha_i \gets (h_i \| \nu_i) \in \zo^{3\lambda}$.
  \item Set $\beta_1 = \dots = \beta_{3\lambda-1} = 1 \in \F$, for a prime field
        $\F$ with $\abs{\F} > C$.
        Set $\beta_{3\lambda} \gets \alpha_i \in \F_{2^n}$.
  \item Prepare a pair of DPF keys:
    \[ (k^{(i)}_{0}, k^{(i)}_{1}) \gets \Gen(\hat \alpha_i, \beta_1, \dots, \beta_{3\lambda})\]
        and send one key to each server, as before.
\end{itemize}

Here, the incremental DPF keys represent secrets
shares of the values of the nodes on
a depth-$(3\lambda+1)$ binary tree.
For client $i$, these node values are all zero except
on the path to leaf $\hat \alpha_i = (\Hash(\alpha_i)\| \nu_i)$, which have
value $1$. 
Finally, the leaf indexed by $\hat \alpha_i \in \zo^{3\lambda}$ contains
the client's full string value $\alpha_i \in \zo^n$, represented
as a field element in $\F_{2^n}$.

Notice that the incremental DPF now operates over strings of length 
$n' = 3\lambda + 1$. Therefore the total key size is:
$\lambda n' + \sum_{\ell=1}^{n'} \log_2 \abs{\F_\ell}$.
We have $\abs{\F_1} = \cdots =\abs{\F_{3\lambda}} \approx C$
and $\abs{\F_{3\lambda+1}} = n$, making the length of each key only 
$\approx 3 \lambda^2 \log_2 C + n$ bits.

The process that the servers use to recover the heavy hitters now changes
slightly: 
\begin{itemize} 
  \item The servers run the $t$-heavy hitters protocol 
          (\cref{proto:heavy}) to find each hash value $h \in \zo^{2\lambda}$ that
        more than $t$ clients submitted.
  \item For each such hash value $h$:
    \begin{itemize}
      \item The servers search for 
        a string $\nu \in \zo^\lambda$ such that at least one client has
        the string $\nu$.
        They can do this using a randomized depth-first search
        variant of \cref{proto:prefix}.
    \item Finally, when then the servers find such a value $\nu$,
        each server publishes the sum of their incremental DPF keys
        evaluated on $(h \| \nu) \in \zo^{3\lambda}$.
        
        If the client who submitted this string is honest, the servers
        will recover a string $\sigma \in \zon$ such that $h = \Hash(\sigma)$.
        Otherwise, the servers retry the randomized depth-first search
        until they find such a string.
        
  \end{itemize}
\end{itemize}

Our extensions for providing differential-privacy (\cref{sec:dp})
are not compatible with this hashing-based technique.
However, they are compatible with a more refined variant
that uses error-correcting codes.

To sketch the idea: We no longer have each client write its
entire $n$-bit string into the leaf of the incremental DPF tree.
Instead, each client encodes its string $\alpha_i \in \zon$ using
an error-correcting code.
Say that the encoded string $E(\alpha_i)$ has $2^k$ symbols,
which we can index by strings $\zo^k$.
Each client picks the index $j \in \zo^k$ of a random symbol
and writes this symbol into the position 
$(\Hash(\alpha_i) \| j) \in \zo^{2\lambda +k}$
in the DPF tree.

 }
\abbr{}{

\section{Cryptographic details}

\subsection{Definition: Incremental DPF}
\label{app:dpfdef}

We seek FSS for the following class of {\em all-prefix point functions}.

\begin{definition}[All-Prefix Point Function] 	\label{def:all-prefix}
We define the class of {\em all-prefix point functions}, each represented by a tuple $(\alpha,(\G_1,\beta_1),\dots,(\G_n,\beta_n))$ (shorthand $(\alpha,\bar \beta)$) where $\alpha \in \zo^n$, and for every $\ell \in [n]$ it holds that $\G_\ell$ is the description of an Abelian group and $\beta_i \in \G_i$, by the function 
	\[ f_{\alpha,\bar{\beta}} : \bigcup_{\ell\in[n]} \zo^\ell \to \bigcup_{\ell \in [n]} \G_\ell, ~~\text{given by} \]
	\[ f_{\alpha,\bar{\beta}}(x_1,\dots,x_\ell) = 	\begin{cases}
									\beta_\ell &\text{if}~(x_1,\dots,x_\ell)=(\alpha_1,\dots,\alpha_\ell)\\
									0 &\text{else}
									\end{cases} \]
\end{definition}

In doing so, we will consider a generalization of standard DPF machinery, endowed with an {\em incremental} evaluation structure wherein each bit of the input $x$ can be incorporated one by one within the DPF evaluation. This will enable us an efficient means for a form of DPF evaluation on input prefixes. The resulting scheme has the same $\Gen$ key-generation syntax and security guarantees as standard DPF. The incremental nature appears in the $\Eval\Next$ and $\Eval\Prefix$ procedures in the place of standard DPF $\Eval$. 

Notationally: in what follows, public values associated with level $\ell$ will be marked with subscript $\ell$; private values (those known or computed only by one party) will receive superscript $i$ and subscripted party id $b$.

\begin{definition}[Incremental DPF: Syntax] \label{def:IDPF-syntax}
A ($2$-party) {\em incremental distributed point function (IDPF)} scheme is a tuple of algorithms $(\Gen,\Eval\Next,\eval\Prefix)$ such that:
  \begin{itemize}
  
  \item $\IDPF.\Gen(1^\lambda, (\alpha,(\G_1,\beta_1),\dots,(\G_n,\beta_n)))$ is a PPT {\em key generation algorithm} that given $1^\lambda$ (security parameter) and a description $(\alpha,(\G_1,\beta_1),\dots,(\G_n,\beta_n))$ of an all-prefix point function, outputs a pair of keys and public parameters $(k_0,k_1, \pp = (\pp_1,\dots,\pp_n))$. We assume that $\pp$ includes the public values $\lambda,n,\G_1,\ldots,\G_n$.
  
  \item $\IDPF.\Eval\Next(b, \st^{\ell-1}_b, \pp_\ell, x_\ell)$ is a polynomial-time {\em incremental evaluation algorithm} that given a server index $b\in\zo$, secret state $\st^{\ell-1}_b$, public parameters $\pp_\ell$, and input evaluation bit $x_\ell \in \zo$, outputs an updated state and output share value: $(\st^\ell_b,y^\ell_b)$.
  
  \item $\IDPF.\eval\Prefix(b,k_b,\pp,(x_1,\dots,x_\ell))$ is a polynomial-time {\em prefix evaluation algorithm} that given a server index $b \in \zo$, key $k_b$, public parameters $\pp$, and input evaluation prefix $(x_1,\dots,x_\ell) \in \zo^\ell$, outputs a corresponding output share value $y^\ell_b$.
  \end{itemize}
\end{definition}

\begin{definition}[IDPF: Correctness and Security] \label{def:IDPF-semantics}
We say that $(\IDPF.\Gen,\IDPF.\Eval\Next)$ as in Definition~\ref{def:IDPF-syntax} is an {\em incremental DPF scheme} if it satisfies the following requirements. 
\begin{itemize}
 
  \item {\bf Correctness:}  
    For every $\lambda,n \in \N$, value $\alpha \in \zo^n$, abelian groups and values $\bar \beta = ((\G_1,\beta_1)\dots,(\G_n,\beta_n))$, level $\ell \in [n]$, and input prefix $(x_1,\dots,x_\ell) \in \zo^\ell$, the following two properties hold.

  \begin{itemize}
  \item $\Eval\Next$. It holds that
  	\[ \Pr[y^\ell_0 + y^\ell_1 = f_{\alpha,\bar\beta}(x_1,\dots,x_\ell)] = 1, \]
where probability is taken over the choice of $(k_0,k_1, \pp = (\pp_1,\dots,\pp_n)) \getsr \Gen(1^\lambda, (\alpha,(\G_1,\beta_1),\dots,(\G_n,\beta_n)))$, and for each $b \in \zo$, $y_b^\ell$ is given by:
	\begin{algorithmic}[1]
	\State $\st_b^0 \gets k_b$;
	\For{($j=1$ to $\ell$)}
	  \State $(\st_b^j, y_b^j) \gets \IDPF.\Eval\Next(b, \st^{j-1}_b, \pp_j, x_j)$;
	\EndFor \\
	\Return $y_b^\ell$
	\end{algorithmic}
	
    \item $\Eval\Prefix$. It holds that
	  	\[ \Pr[y^\ell_0 + y^\ell_1 = f_{\alpha,\bar\beta}(x_1,\dots,x_\ell)] = 1, \]
where probability is taken over the choice of $(k_0,k_1, \pp = (\pp_1,\dots,\pp_n)) \getsr \Gen(1^\lambda, (\alpha,(\G_1,\beta_1),\dots,(\G_n,\beta_n)))$, and for each $b \in \zo$, $y_b^\ell \gets \Eval\Prefix(b,k_b,\pp,(x_1,\dots,x_\ell))$.
		
    \end{itemize}
	
  \item {\bf Security:} For each $b\in\zo$ there is a PPT algorithm $\Sim_b$ (simulator), such that for every sequence $((\alpha,\bar\beta)_\lambda)_{\lambda\in\N}$ of polynomial-size all-prefix point function descriptions and polynomial-size input sequence $x_\lambda$,
  the outputs of the following experiments $\Real$ and $\Ideal$ are computationally indistinguishable:
  \begin{itemize}
  \item 
$\Real_\lambda$: $(k_0,k_1,\pp) \getsr \Gen(1^\lambda, (\alpha,\bar\beta)_\lambda)$; Output $(k_b,\pp)$.

\item 
  $\Ideal_\lambda$: Output $\Sim_b(1^\lambda, (n,\G_1,\dots,\G_n))$.
  \end{itemize}
  \end{itemize}
\end{definition}

\abbr{}{
\subsection{Instantiation and optimizations for our incremental DPF construction}\label{app:opts}

 \itpara{Instantiating PRG via AES.} Following~\cite{Splinter}, the length-doubling PRG $G$ can be instantiated via two executions of {\em fixed-key} AES (taking e.g.\ $\lambda = 127$), using AES-NI hardware instructions for AES encryption. Evaluation via $\Eval\Next$ or $\Eval\Prefix$ thus requires just one fixed-key AES encryption per level, as only one half of each expanded PRG output is relevant for a given input. 
  
For example, for $m$-bit output groups $\G_\ell$, this results in the following costs, in units of fixed-key AES encryptions:
    \begin{itemize}
    \item $\Gen$: $4n \left( 1 +  \lceil m / \lambda \rceil \right)$ 
    \item $\Eval\Next(\ell)$: $1 + \lceil m/ \lambda \rceil$ 
    \item $\Eval\Prefix(\ell)$: $\ell( 1 + \lceil m/ \lambda \rceil)$ 
    \end{itemize}
  
\itpara{Recovering DPF.} Taking $\G_\ell = \{\bot\}$ for $\ell \in \{1,\dots,n-1\}$ in the above construction, i.e.\ $\lceil \log |\G_\ell| \rceil = 0$ for all but the final level $\G_n$, we recover the DPF construction and complexity of~\cite{BGI16-FSS}. In this sense, our incremental DPF construction is a strict generalization.

\itpara{Subtractive shares.} If for an application it suffices to produce {\em subtractive}  shares of the output, i.e., for which $y_0 - y_1 = \beta$ as opposed to $y_0+y_1 = \beta$, then the group inverse computation, denoted by multiplication by $(-1)^b$ in line~\ref{step:Eval-W2} of $\Eval\Next$ can be removed.

\itpara{PRG evaluation optimization.} For the case of small output groups $\G_\ell$, instantiating the pseudo-random expansion $\Convert_{\G'_\ell}: \zo^\lambda \to \zo^\lambda \times \G_\ell$ via an execution of the PRG $G : \zo^\lambda \to \zo^{2\lambda+2}$ is wasteful. Instead, this expansion can be absorbed into the next-level execution of $G$, ``stealing'' a portion of the pseudo-random output bits of $G$ to be interpreted as a pseudorandom element of $\G_\ell$. 
  
  For example, implementing $G$ via 2 AES encryptions as described above, but fixing two bits of input, can be viewed as a pseudo-random generator $\zo^{126} \to \zo^{2(126)+2} \times \zo^2$, at the expense of slightly decreased security parameter. The $\zo^{2(126)+2}$ component of the output can be used as required for the remainder of the next-level execution, and the $\zo^2$ portion can be interpreted as a representation of a pseudo-random $\G_\ell$ element (e.g., if $\G_\ell = \Z_4$).
Continuing the running example, e.g.\ for the case of 1- or 2-bit output groups $\G_\ell$, with this optimization results in the following costs, in units of fixed-key AES encryptions:
    \begin{itemize}
    \item $\Gen$: $4n$
    \item $\Eval\Next(\ell)$: $1$
    \item $\Eval\Prefix(\ell)$: $\ell$
    \end{itemize}
}

\begin{figure}
\begin{framed}
{\bf Incremental DPF $(\Gen,\Eval\Next,\Eval\Prefix)$}\\
Let $\G' := \zo^{\lambda} \times \G$, where $\zo^\lambda$ has bitwise addition. \\
  Let $G : \{0,1\}^\lambda \to \{0,1\}^{2\lambda+2}$ and $\Convert_{\G'} : \zo^\lambda \to \G'$ be pseudorandom generators\abbr{.}{ (see Figure~\ref{fig:ConvertFn}).}

  \vspace{.1in}

$\Gen(1^\lambda, (\alpha,(\G_1,\beta_1),\dots,(\G_n,\beta_n)))$:
  \begin{algorithmic}[1]
  \State Let $\alpha = \alpha_1,\ldots,\alpha_n \in \{0,1\}^n$ be the bit decomp of $\alpha$
  \State Sample random $s^{(0)}_0 \getsr \{0,1\}^\lambda$ and $s^{(0)}_1 \getsr \{0,1\}^\lambda$
  \State Let $t^{(0)}_0 \gets 0$ and $t^{(0)}_1 \gets 1$  \label{step:t} \For {$\ell = 1$ to $n$}	
    \State $s^L_b||t^L_b ~\big|\big|~s^R_b || t^R_b  \gets G(s^{(\ell-1)}_b)$ for $b=0,1$
\If {$\alpha_{\ell}=0$} $\Keep \gets L$, $\Lose \gets R$
    \Else ~$\Keep \gets R$, $\Lose \gets L$
    \EndIf 
    \State $s_{CW} \gets s^\Lose_0 \oplus s^\Lose_1$
    \State $t^L_{CW} \gets t^L_0 \oplus t^L_1 \oplus \alpha_{\ell} \oplus 1$ and
		$t^R_{CW} \gets t^R_0 \oplus t^R_1 \oplus \alpha_{\ell}$
    \State $t^{(\ell)}_b \gets t^\Keep_b \oplus t^{(\ell-1)}_b \cdot t^\Keep_{CW}$ for $b=0,1$
    \State $\tilde s^{(\ell)}_b \gets s^\Keep_b \oplus t^{(\ell-1)}_b \cdot s_{CW}$ for $b = 0,1$
\State $s_b^{(\ell)} || W^{(\ell)}_b \gets \Convert_{\G'_\ell} (\tilde s_b^{(\ell)})$ for $b=0,1$  \label{step:Gen-W1}
    \State $W^{(\ell)}_{CW} \gets (-1)^{t_1^{(\ell)}} \cdot [ \beta_\ell - W^{(\ell)}_0 + W^{(\ell)}_1]$  \label{step:Gen-W2}
    \State $CW^{(\ell)} \gets s_{CW} || t^L_{CW} || t^R_{CW} || W^{(\ell)}_{CW}$ 
\EndFor
\State Let $k_b \gets s^{(0)}_b$ for $b=0,1$.
	\State Let $\pp \gets CW^{(1)}, \cdots, CW^{(n)}$ \\
  \Return $(k_0,k_1,\pp)$
  \end{algorithmic}

  \vspace{.1in}
	
$\Eval\Next(b,\st^{\ell-1}_b, \pp_\ell = CW^{(\ell)}, x_\ell)$:
  \begin{algorithmic}[1]
    \State Parse $\st^{\ell-1}_b = s^{(\ell-1)} || t^{(\ell-1)}$
\State Parse $CW^{(\ell)} = s_{CW} || t^L_{CW} || t^R_{CW} || W_{CW}$
		\State Parse $G(s^{(\ell-1)}) = \hat{s}^L|| \hat{t}^L ~\big|\big|~\hat{s}^R || \hat{t}^R $
    \State $\tau^{(\ell)} \gets 
    	(\hat{s}^L||\hat{t}^L ~\big|\big|~\hat{s}^R || \hat{t}^R )  \oplus \newline 
	~~~~(t^{(\ell-1)} \cdot \big[s_{CW} || t^L_{CW} || s_{CW} || t^R_{CW} \big])$
    \State Parse $\tau^{(\ell)} = s^L || t^L ~\big|\big|~ s^R || t^R  \in \{0,1\}^{2\lambda+2}$
    \If {$x_\ell = 0$}  $\tilde s^{(\ell)} \gets s^L, t^{(\ell)} \gets t^L$ 
    \Else ~$\tilde s^{(\ell)} \gets s^R$, $t^{(\ell)} \gets t^R$ 
    \EndIf
    \State $s^{(\ell)} || W^{(\ell)} \gets \Convert_{\G'_\ell} (\tilde s^{(\ell)})$ \label{step:Eval-W1}
\State $\st^\ell \gets s^{(\ell)} || t^{(\ell)}$
    \State $y_b^\ell \gets (-1)^b \cdot [W^{(\ell)} +  {t^{(\ell)}} \cdot W_{CW}]$ \label{step:Eval-W2}
\\
		\Return $(\st^\ell, y_b^\ell )$
  \end{algorithmic}

\vspace{.1in}

$\Eval\Prefix(b, k_b, \pp,(x_1,\dots,x_\ell))$:
  \begin{algorithmic}[1]
	\State Let $s^{(0)}=k_b$ and $t^{(0)} = b$.
  \State Parse $\pp = CW^{(1)}, \cdots, CW^{(n)}$.
  \State $\st^0_b \gets s^{(0)}||t^{(0)}$
  \For {$j=1$ to $\ell$}
    \State $(\st_b^j,y_b^j) \gets \IDPF.\Eval\Next(b,\st^{j-1}_b,CW^{(j)},x_j)$
  \EndFor\\
  \Return $y^\ell_b$ 
  \end{algorithmic}
\end{framed}

\caption{Pseudocode for incremental DPF construction. 
\abbr{}{Pseudocode for $\Convert_\G$ is given in Figure \ref{fig:ConvertFn}.}
The symbol $||$ denotes string concatenation. Subscripts 0 and 1 refer to party id. All $s$ values are $\lambda$-bit strings, $W$ values are elements in $\G_\ell$, which are represented in $\lceil \log |\G_\ell| \rceil$ bits, and $t$ and $b$ values are single bits.}
\label{fig:IDPF}
\end{figure}

\abbr{}{
\begin{figure}[h]
\begin{framed}

$\Convert_{\mathbb{G}}(s)$:
  \begin{algorithmic}[1]
	\State Let $u\gets |\mathbb{G}|$.
	\If {$u=2^m$ for an integer $m$}
    \State Return the group element represented by\\
    \quad\quad $G(s)$, for a PRG $G:\{0,1\}^\lambda \rightarrow \{0,1\}^{m}$.
	\Else 
    \State Let $\ell \gets \lceil \log_2 u\rceil + \lambda$.
		\State Return the group element represented by \\
    \quad \quad $G(s) \bmod u$, for a PRG $G:\{0,1\}^\lambda \rightarrow \{0,1\}^{\ell}$.
	\EndIf
 \end{algorithmic}
\end{framed}
\caption{Pseudocode for converting a string $s \in \{0,1\}^\lambda$ to an element in a group $\mathbb{G}$. If $s$ is random then $\Convert_{\mathbb{G}}(s)$ is pseudo-random.}
\label{fig:ConvertFn}
\end{figure}
}

\subsection{Proof of \cref{prop:IDPF}}\label{app:idpf}

\newcommand{\RemainingCW}{{\sf RemainingCW}}
\begin{proof}[Proof of Proposition~\ref{prop:IDPF}]

Security: We prove that each party's key $k_b$ is pseudorandom. This will be done via a sequence of hybrids, where in each step another correction word $CW^{(\ell)}$ within the key is replaced from being honestly generated to being random, for $\ell=1$ to $n$.

The argument for security goes as follows. Each party $b \in \zo$ begins with a random seed $s^{(0)}_b$ that is completely unknown to the other party. In each level of key generation (for $\ell=1$ to $n$), the parties apply a PRG to their seed $s^{(\ell-1)}_b$ to generate 4 items: namely, 2 seeds $\tilde s_b^L, \tilde s_b^R$, and 2 bits $t_b^L,t_b^R$. This process will {\em always} be performed on a seed which appears completely random and unknown given the view of the other party; because of this, the security of the PRG guarantees that the 4 resulting values appear similarly random and unknown given the view of the other party. 

The $s_{CW} || t^L_{CW} || t^R_{CW}$ portion of the  $\ell$th level correction word $CW^{(\ell)}$  ``uses up'' the secret randomness of 3 of these 4 pieces: the two bits $t_b^L,t_b^R$, and the seed $\tilde s_b^\Lose$ for $\Lose \in \{L,R\}$ corresponding to the direction {\em exiting} the ``special path'' $\alpha$ (i.e.\ $\Lose = L$ if $\alpha=1$ and $\Lose = R$ if $\alpha=0$). However, given this $CW^{(\ell)}$, the remaining seed $\tilde s^\Keep_b$ for $\Keep \neq \Lose$ still appears random to the other party. 
This seed $\tilde s^\Keep_b$ is expanded to $s_b^{(\ell)} || W^{(\ell)}_b$, again appearing random to the other party. The final portion $W_{CW}$ of the $\ell$th level correction word $CW^{(\ell)}$ ``uses up'' the secret randomness of the $W^{(\ell)}$, leaving $s_b^{(\ell)}$ that appears random to the other party.

The argument then continued in similar fashion to the next level, beginning with seeds $s_b^{(\ell)}$.
\end{proof}

\subsection{Malicious-secure sketching}
\label{app:mal-sketch}

\paragraph{Protocol description.}
The client holds vector $\bar v$. The servers hold common randomness $\bar r = (r_1,\dots,r_m)$ and $\bar r^* = (r_1^2,\dots,r_m^2)$.

\begin{enumerate}
\item Client samples random $\kappa \gets \F$ and sends:
  \begin{itemize}
  \item DPF shares of $(\bar v, \kappa \bar v)$ 
  \item Correlation for aiding servers' secure computation; Namely,
  	additive shares (over $\F$) of: 
  	\begin{enumerate}
	\item Random $(a,b,c) \in \F^3$. These can be emulated by sending a single random PRG seed to each server, amortizing across the number of levels. 
	\item $A := [-2a + \kappa]$
	\item $B := [(a^2+b) +(-a\kappa +c)]$
	\end{enumerate}
  \end{itemize}
Altogether, the {\bf client sends (amortized) 2 $\F$-elements plus one DPF key} (with payload $\zo \times \F$) {\bf to each server}.  

\item Servers compute sketch. Each server $\sigma \in \zo$:

Evaluates his DPF key on all elements of the domain $[m]$; denote the resulting vector by $(\bar v_\sigma, \bar v^*_\sigma) \in \zo^m \times \F^m$.

Computes the following 3 field elements:
\begin{align*}
  z_\sigma &\gets \langle \bar r, \bar v_\sigma \rangle;& z^*_\sigma &\gets \langle \bar r^*, \bar v_\sigma \rangle;
  &z^{**}_\sigma &\gets \langle \bar r,   \bar v_\sigma^* \rangle.
\end{align*}

    \item Round 1: Servers exchange masked input values. 
    
    Each server $\sigma \in \zo$ sends: 
    	\[ (z_\sigma + a_\sigma),~ (z^*_\sigma + b_\sigma), ~(z^{**}_\sigma + c_\sigma),\] 
where $a_\sigma,b_\sigma,c_\sigma$ are his shares of $a,b,c$.
    
    {\bf Communication: 3 $\F$-elements per server.}
    
    Locally: Compute $Z = (z_0+a_0)+(z_1+a_1)=(z+a)$. Compute analogous $Z^*, Z^{**}$.

     Locally use these (public) $Z,Z^*, Z^{**}$, the PRG-defined shares of $(a,b,c)$, and client-provided shares of $A,B$ to homomorphically derive additive shares of the following degree-2 polynomial:    
    \begin{align*}
     (z^2 - &z^*) + (\kappa \cdot z - z^{**}) \\
    	&= \left[ (Z - a)^2 - (Z^*-b) \right] +  \left[ (Z-a)\kappa - (Z^{**} - c) \right] \\
    	&= [Z^2] - [Z^*] - [Z^{**}] + A[Z] + B 
    \end{align*}

(Note that each term in square brackets is publicly computable, and each coefficient is held additively secret shared by the servers.)
    
    \item Round 2: Exchange evaluated shares: Each server sends their share of the above to the other server.
    
    {\bf Communication: 1 $\F$-element per server.}
    
    Locally: Combine the shares. If the sum is nonzero, abort; otherwise, accept.

{\bf Total Comm: 2 rounds, 4 $\F$-elements per server.}
  
\end{enumerate}

\abbr{In the full version of this work~\cite{full}, we prove that the above protocol
protects security against both a malicious client and a malicious server.}{

\subsubsection{Security analysis}

We begin by showing that the use of client-provided correlated randomness (shares of $a,b,c,A,B$) as well as the extra verification checks does not adversely affect the guarantees of the Boyle et al.~\cite{BGI16-FSS} protocol against a malicious client.

\begin{claim}[Malicious client] 
Suppose $\bar v$ as defined by the server's shares $\bar v_0,\bar v_1$ is {\em not} a legal vector $\alpha \bar e_i$ for some $i \in [m]$ and $\alpha \in \zo$.  
Then the client will be rejected except with probability bounded by $2/|\F|$.
\end{claim}
\begin{proof} 
The proof follows the argument as in Boyle et al.~\cite{BGI16-FSS}.
In our case, a malicious client has the ability to send arbitrary maliciously chosen values for: $\bar v, \bar v^*$ (supposed to be $\kappa \bar v$), $A, B$. 

Consider the expression evaluated by the servers given these values, expressed as a polynomial in the variables $r_1,\dots,r_m$:
    \begin{align*}
      [Z^2& - Z^*] - [Z^{**}] + A[Z] + B \\
    	&= \left[(\Sigma r_iv_i + a)^2 - (\Sigma r_i^2 v_i + b)\right] - [\Sigma r_i v^*_i + c]  \\
		&\hspace{2in} + A[\Sigma r_iv_i + a] + B \\
	&= \Sigma r_i^2(v_i^2-v_i) + \Sigma_{i\neq j} r_ir_j (v_iv_j) + \Sigma r_i (X_i) + Y,
    \end{align*}  
for some terms $X_i,Y$ that do not contain any $r_i$. By the Schwartz-Zippel Lemma, if the above polynomial is not the 0 polynomial, then over a random choice of the variables $r_1,\dots,r_m$, the polynomial will evaluate to 0 with probability no greater than $2/|\F|$. Thus to succeed with greater probability within the verification check, the adversary must select offsets for which the coefficient of each monomial of the respective polynomials is set to 0. 

In particular, the coefficient of each $r_i r_j$ for $i \neq j$ requires $v_iv_j = 0$, and thus $\forall i \neq j, v_iv_j = 0$. This implies $\bar v$ can have at most one nonzero entry. Further, the coefficient of $r_i^2$ is $(v_i^2-v_i)$, requiring $v_i \in \{0,1\}$ for each element $v_i$ of $\bar v$. Combined, these together imply that $\bar v$ is of the required form.
\end{proof}

\newcommand{\view}{{\sf view}}

We now prove that the above protocol guarantees client privacy against a malicious server.

\begin{proposition}[Malicious server]
For every malicious server $\tilde S$ there exists a simulator $\Sim$ for which the view $\view_{\tilde S}(\bar v)$ of $\tilde S$ in execution of the protocol on {\em honest-client input} is indistinguishable to the output of $\Sim$.
\end{proposition}

\begin{proof} 
The client-aided two-party secure computation protocol used (i.e., client-supplied correlated randomness, combined with Rounds 1 \& 2) is secure against a malicious server, up to additive offsets to inputs and outputs of the computation~\cite{GenkinIPST14,BGI16-FSS,BGI19}. Note that an additive offset to the {\em output} is irrelevant for client privacy (recall that we do not address robustness of the computation against a malicious server).

Consider then the effect of maliciously selected additive offsets to the {\em inputs} $z,z^*,z^{**}$ of the secure computation. Because of the random secret mask values $a,b,c$, the adversary's offsets $\Delta,\Delta^*,\Delta^{**}$ must be selected independently of the true values of $z,z^*,z^{**}$.

An honest client implies its corresponding vector $\bar v$ is of legal form, and thus $z=r_i$ for some $i \in \{0\} \cup[m]$, where $r_0 := 0$ for notational simplicity. Similarly, $z^* = r_i^2 v_i$, and $z^{**} = \kappa r_i$. Consider the resulting output computed within the secure computation on the corresponding offset inputs:
    \begin{align*}
    [(&z+\Delta)^2 - (z^* + \Delta^*)] +  [ \kappa(z + \Delta) - (z^{**} + \Delta^{**})] \\
      &= [(r_i+\Delta)^2 - (r_i^2 + \Delta^*)] +  [ \kappa(r_i + \Delta) - (\kappa r_i + \Delta^{**})] \\
      &= [2r_i \Delta + \Delta^2 - \Delta^*] +  [\kappa \Delta - \Delta^{**} ] \\
      &=  \kappa[\Delta] + [2r_i \Delta + \Delta^2 - \Delta^*- \Delta^{**}]  
    \end{align*}
If it is the case that $\Delta \neq 0$, then the above expression is uniformly distributed over the client's random (secret) choice of $\kappa$. On the other hand, if $\Delta = 0$, then the potentially sensitive contribution $2r_i\Delta$ is removed, and the resulting expression $(-\Delta^* - \Delta^{**})$ is fully simulatable.  

This gives rise to the following simulator. 
  \begin{enumerate}
  \item $\Sim$ sends values on behalf of the Client: (a) pseudorandom shares of $(a,b,c) \in \F^3$, (b) random $\F$-elements in the place of additive shares of $A,B$, and (b) a DPF share generated for an arbitrary input in the place of $(\bar v,\kappa \bar v)$.
  \item $\Sim$ sends random values on behalf of the honest server in Round 1 in the place of shares of the masked inputs $Z=(z+a),Z^*=(z^*+b),Z^{**}=(z^{**}+c)$. 
  \item $\Sim$ computes the values the malicious server {\em should} have sent in Round 1 (corresponding to his shares of $Z,Z^*,Z^{**}$), as a function of the received simulated values from the previous two steps and $\bar r,\bar r^*$. Given the values the malicious server {\em did} send in Round 1, denote the effective additive offsets to the correct values as $(\Delta,\Delta^*,\Delta^{**}) \in \F^3$. 
    \begin{itemize}
    \item If $\Delta \neq 0$, Simulate the Round 2 message of the honest server with a random $\F$ element. 
    \item If $\Delta = 0$, then simulate as the appropriate additive share of the output value $(-\Delta^*-\Delta^{**})$. 
    \end{itemize}
  \end{enumerate}
\end{proof}
}

 }
\abbr{}{

\section{Extractable DPF}\label{app:taint}\label{app:extractx}
In previous work \cite{GI14,BGI15,BGI16-FSS} a DPF scheme was defined as a pair of algorithms $(\Gen,\Eval)$, such that $\Gen$ takes as input a security parameter and outputs a pair of keys $(k_0,k_1)$, while $\Eval$ takes as input a key and an input point $x$ and outputs a group element. In the DPF constructions of these papers, an honest execution $\Gen$ results in two output keys that have a shared portion. These constructions are not extractable since a malicious client can generate two keys in which this part is not identical, and thereby control the output value of two locations instead of the output at just a single point. 

We use an  alternative formulation for DPF, which we call {\em DPF with public parameters}. The definition separates the keys into two private parts $(k_0,k_1)$ and a public part $\pp$, similarly to Definitions \ref{def:IDPF-syntax} and \ref{def:IDPF-semantics} for IDPF. A party running the $\Eval$ algorithm takes as input
a full key $(k_b,\pp)$ for $b=0,1$.

\begin{definition}[DPF with public parameters: Syntax] \label{def:alt-DPF}
A ($2$-party) {\em distributed point function (DPF)} scheme 
is a pair of algorithms $(\Gen,\Eval)$ such that:
  \begin{itemize}
  \item $\Gen(1^\lambda, \alpha,(\G,\beta))$ is a PPT {\em key generation algorithm} that given $1^\lambda$ (security parameter) and a description $\alpha,(\G,\beta)$ of a point function, where $\alpha\in\zo^n$, $\G$ is an Abelian group and $\beta\in\G$, outputs a pair of keys and public parameters $(k_0,k_1, \pp)$. We assume that $\pp$ determines the public values $\lambda,n,\G$. 
  
  \item $\Eval(b,k_b,\pp,x)$ is a polynomial-time {\em evaluation algorithm} that given a server index $b \in \zo$, key $k_b$, public parameters $\pp$, and input $x \in \zon$, outputs a corresponding output share value $y_b$.
  \end{itemize}
\end{definition}

The correctness and security properties of DPF with public parameters are essentially identical to the analogous  properties of IDPF in Definition \ref{def:IDPF-semantics}.

We now formally define the basic notion of extractable DPF. The following definition extends Definition~\ref{def-edpf-simple} by allowing the adversary to pick an arbitrary ``sparse'' payload subset $P\subset\G \setminus \{0\}$, which can be represented by an efficient circuit. We call $P$ the set of {\em  permissible outputs}.

\begin{definition}[Extractable DPF]
\label{def:cr}
We say that a DPF scheme in the random-oracle model is {\em extractable} if there is a PPT extractor $E$,
such that every PPT adversary $A$ wins the following game with negligible probability in the security parameter $\lambda$, 
taken over
the choice of a random oracle~$G$ and the secret random coins of $A$.

\begin{itemize}
\item $(1^n,\G,P) \leftarrow A(1^\lambda)$, where $\G$ is an Abelian group of size $|{\mathbb G}|\ge 2^\lambda$ and $P \subseteq \G \setminus \{0\}, |P| \le 2^{\lambda/3}$, is represented by a circuit $P:\G\to\zo$.
\item $(k^*_0,k^*_1,\pp^*,x^*) \leftarrow A^G(1^\lambda,1^n,\G,P)$, where $x^* \in \zon$, and $G$ is a random oracle. We assume that $\pp^*$ determines the correct public values $(1^\lambda,1^n,\G)$. 
\item $x \leftarrow E(k^*_0,k^*_1,\pp^*,P,T)$, where $x \in \zon$ and $T=\{ q_1,\ldots,q_t \}$ is the transcript of $A$'s $t$ oracle queries. 
\end{itemize}

$A$ wins the game 
if $x^* \neq x$ and $\Eval^G(0,k^*_0,pp^*,x^*)+$ $\Eval^G(1,k^*_1,pp^*,x^*)$ $\in P$.  
\end{definition}

The size restrictions on $\G$ and $P$ are used to simplify Definition~\ref{def:cr}. Lemma \ref{lm:taint} below analyzes the extractability property of our DPF and IDPF constructions for more general $\G$ and $P$. 

Definition \ref{def:cr} can be generalized in a natural way to extractable IDPF.
\begin{definition}[Extractable IDPF]
\label{def:cri}
We say that an IDPF scheme in the random-oracle model is {\em extractable} if there is a PPT extractor $E$,
such that every PPT adversary $A$ wins the following game with negligible probability in the security parameter $\lambda$, 
taken over
the choice of a random oracle~$G$ and the secret random coins of $A$.

\begin{itemize}
\item $(1^n,(\G_1,P_1),\ldots,(\G_n,P_n)) \leftarrow A(1^\lambda)$, where it holds for all $i=1,\ldots,n$ that $\G_i$ is an Abelian group of size $|{\mathbb G_i}|\ge 2^\lambda$ and $P_i \subseteq \G_i \setminus \{0\}, |P_i| \le 2^{\lambda/3}$, is represented as a circuit $P_i:\G \to \{0,1\}$.
\item $(k^*_0,k^*_1,\pp^*,x^*) \leftarrow A^G(1^\lambda,1^n,(\G_1,P_1),\ldots,(\G_n,P_n))$, where $x^* \in \{0,1\}^j, 1 \le j \le n$, and $G$ is a random oracle. We assume that $\pp^*$ includes the public values $(1^\lambda,1^n,\G_1,\ldots,\G_n)$. 
\item $x \leftarrow E(k^*_0,k^*_1,\pp^*,P_1,\ldots,P_n,T)$, where $x \in \zon$ and $T=\{ q_1,\ldots,q_t \}$ is the transcript of $A$'s $t$ oracle queries. 
\end{itemize}

$A$ wins the game if $x^* \neq x^{|x^*|}$, for a prefix $x^{|x^*|}=(x_1,\ldots,x_{|x^*|})$ of $x$, and $\IDPF.\eval\Prefix^G(0,k^*_0,\pp^*,x^*)+\IDPF.\eval\Prefix^G(1,k^*_1,\pp^*,x^*) \in P_{|x^*|}$.  
\end{definition}

\begin{notation}
Let $\G$  be a group with group action $+$ and let $P \subseteq \G \setminus \{0\}$. We denote by $P^-$ the set of all differences of elements in $P$, i.e. $P^-=\{a \in \G~|~\exists p,p' \in P, p+(-p')=a\}$. For $n$ pairs $(\G_1,P_1),\ldots,(\G_n,P_n)$, let $\rho= \max \left\{\frac{|P^-_1|}{\G_1},\ldots,\frac{|P^-_n|}{\G_n}\right\}$.
\end{notation}

Some useful examples of $P$ and $P^-$ include any set $P$ such that $|P|=1$, in which case $|P^-|=1$, and any set $P$ which is an interval, i.e. $\exists~a,g \in \G$ such that $P=\{a+i \cdot g~|~0\le i <|P|-1\}$, in which case $|P^-| \le 2|P|$. For a general subset $P$ it holds that $|P^-| \le |P|^2$.

\begin{lemma}
\label{lm:taint}
The scheme constructed in Section \ref{sec:inc} is an extractable IDPF scheme and the DPF scheme constructed in \cite{BGI16-FSS} is an extractable DPF  scheme. Moreover, the probability of an adversary $A$ winning the security game in either scheme is at most $\epsilon_A = \left(4t^2+2nt+1 \right) \left(\max \left\{\rho,\frac{1}{2^\lambda}\right\}\right)$.
\end{lemma}

\begin{proof}
Let $A$ be an adversary, let its output be $(1^n,(\G_1,P_1),\ldots,(\G_n,P_n)),(k^*_0,k^*_1,\pp^*,x^*)$, and let $T=\{ q_1,\ldots,q_t \}$ be the transcript of its oracle queries. Recall that for each string $x^i=x_1,\ldots,x_i \in \{0,1\}^i$ the algorithm $\IDPF.\Eval\Prefix(b, k^*_b, \pp^*, x^i)$ returns $y^i_b$ as output and its last internal state is $\st_b^i$. 

The extractor algorithm $E$ is a restriction of the following algorithm $E'$, which may run in super-polynomial time. $E'$ assigns a value to each string $x \in \bigcup_{i=1}^n\{0,1\}^i$ based on the transcript $T$. $E'$ runs $\IDPF.\Eval\Prefix^G(b, k_b,\pp^*, x)$ for $b=0,1$ and stores all the oracle calls it made to $G$ in $T_x$. If $T_x \subseteq T$ then based on the oracle calls that the adversary made the adversary can evaluate the output of the IDPF on  $x$. In this case, $E'$ assigns to $x$ the value $\IDPF.\Eval\Prefix^G(0, k_0^*,\pp^*, x)+\IDPF.\Eval\Prefix^G(1, k_1^*,\pp^* x)$. Otherwise, $E'$ implicitly assigns to $x$ the value $0$. Note that if $E'$ assigns $0$ to $x$ then it also assigns $0$ to any $x'$ such that $x$ is a prefix of $x'$.

The reason that $E'$ may not run  in polynomial  time is that depending on $G$, the oracle queries in $T$ could be sufficient to evaluate a large number of strings in $\bigcup_{i=1}^n\{0,1\}^i$, possibly many more than $t$. $E$ avoids this problem by limiting the number of identical queries it analyzes for each string length to  one. Finally, $E$ chooses for each level $i$, a string $x^i$ that has a value in $P_i$, or an arbitrary string (the string of all $1$ bits) if all the values it assigned are not in $P_i$. Pseudo-code for $E$ appears in Figure \ref{fig:tainted}.

\begin{figure}
\begin{framed}

$E(k^*_0,k^*_1,\pp^*,P_1,\ldots,P_n,T)$
  \begin{algorithmic}[1]
	\State Extract $\lambda,n, \G_1,\ldots,\G_n$ from $\pp^*$.
  \State Let  $C \leftarrow \{(\varepsilon,0)\}$, $\bar{C} \leftarrow \emptyset$, and $ST \leftarrow \emptyset$
	\While {$C \neq \emptyset$}
		\State Let $(x,v) \in C$
		\For {$z = 0,1$}
			\For	{$b = 0,1$}
				\State $y_b^{|x|+1} \leftarrow (\IDPF.\Eval\Prefix)^G(b, k_b^*, \pp^*,x||z)$ \label{ln:evalp}
				\State Let $L_{b,z} \leftarrow (s_b^{(|x|+1)},\tilde{s}_b^{(|x|+1)})$
			\EndFor
			\If {$\exists ((s,\tilde{s}),|x|+1) \in ST$ s.t. $s = s_b^{(|x|+1)}$} \label{ln:ifst}
				\State $Stop \leftarrow 1$
			\Else ~$Stop \leftarrow 0$
			\EndIf
			\If {$L_{0,z},L_{1,z} \subseteq T$ and  (Stop=0)} \label{ln:ifl}	
				\State $C \leftarrow C \bigcup \{(x||z,y_0^{|x|+1}+y_1^{|x|+1})\}$
				\State $\bar{C} \leftarrow \bar{C} \bigcup C$
				\State $ST \leftarrow ST \bigcup \{(L_{0,z},|x|+1),(L_{1,z},|x|+1)\}$
			\EndIf
		\EndFor
		\State $C \leftarrow C \setminus \{(x,v)\}$
	\EndWhile
	\For {$i = 1$ to $n$}
		\State Let $X^i \leftarrow \{x~|~(x,v) \in\bar{C}, |x|=i, v \in P_i\} \bigcup \{1^i\}$
		\State Let $x^i$ be smallest string lexicographically in $X^i$
	\EndFor \\ 
  \Return $(x^1,\ldots,x^n)$
  \end{algorithmic}

\end{framed}

\caption{Pseudocode for extractor algorithm $E$. 
The symbol $\varepsilon$ denotes the empty string. The set $C$ includes pairs $(x,v)$ of  the input strings $x$ that the algorithm intends to examine together with $v$, the sum of the evaluation of the two keys on  $x$. The set $\bar{C}$ includes all the pairs $(x,v)$ that the algorithm examines throughout its execution. The set $ST$ includes all the queries to the oracle $G$ that are made by $(\IDPF.\Eval\Prefix)^G$. A subscript $b \in  \{0,1\}$ refers to party id, while $z \in \{0,1\}$ refers to a bit in the input string. The strings $s_b^{(|x|+1)},\tilde{s}_b^{(|x|+1)}$ are determined by the execution in Line \ref{ln:evalp} of $\IDPF.\Eval\Prefix$. The flag $Stop$ is used to halt the execution if two identical calls to the oracle are made for inputs of equal length. The sequences $L_{b,z}$ are viewed as sets in Line \ref{ln:ifl}. The string $1^i$ is the all-one $i$-bit string.}
\label{fig:tainted}
\end{figure}

$E$ runs in polynomial time  since the number of different strings it examines is at most $2nt$ and it performs at most linear time work in the length of each string. To prove the bound on the number of strings, recall that the oracle queries that $\IDPF.\Eval\Prefix$ makes are on strings $s_b^{(\ell)}$ and $\tilde{s}_b^{(\ell)}$. There are  at most $t$ such distinct queries, and $E$ adds each query together with the length of the associated input string $x$ to  the set $ST$. The if statements in Lines \ref{ln:ifst} and \ref{ln:ifl} ensure that  a query $s_b^{(\ell)}$ appears only once per possible string length in $ST$ and therefore at most $n$ times. $E$ adds a string $x$ to $C$ only  if it adds a tuple that includes $s_b^{(\ell)}$ to $ST$, i.e. at most $nt$  strings are added to $C$ over the course of the algorithm. $E$ examines the two extensions $x||z, z \in \{0,1\}$ for each $x \in C$, and therefore examines at most $2nt$ strings.

We use the following notation for the string $x^*$ that $A$ outputs and the strings $x^i$ that $E$ outputs. Let $\ell=|x^*|$, let $x^{*(i)}$ be an $i$-bit prefix of $x^*$, denote the oracle queries associated with the last bit of $x^{*(i)}$ by $s^{*(i)}_b$ and $\tilde{s}^{*(i)}_b$, and denote the output of $(\IDPF.\Eval\Prefix)^G(b, k_b^*, \pp^*,x^{*(i)})$ by $y^{*i}_b$, for $b \in \{0,1\}$. For any $x^i$ output by $E$, denote the oracle queries by $s^{(i)}_b$ and $\tilde{s}^{(i)}_b$, and denote the output of $(\IDPF.\Eval\Prefix)^G(b, k_b^*, pp^*, x^i)$ by $y^i_b$.

We separate the analysis of the probability that the client outputs $x^* \neq x^\ell$ such that $(\IDPF.\Eval\Prefix)^G(0, k_0^*, \pp^*,x^*)+(\IDPF.\Eval\Prefix)^G(1, k_1^*, \pp^*, x^*) \in P_\ell$ into three cases.
\begin{enumerate}
\item There exist some $1 \le i \le \ell, b \in \{0,1\}$ such that $s_b^{*(i)}$ or $\tilde{s}_b^{*(i)}$ is not a query in $T$.
\item $s_b^{*(i)}$ and $\tilde{s}_b^{*(i)}$ are queries in $T$ for  all $i$ and $b$, and $x^* \in \bar{C}$.  
\item $s_b^{*(i)}$ and $\tilde{s}_b^{*(i)}$ are queries in $T$ for  all $i$ and $b$, and $x^* \not \in \bar{C}$. 
\end{enumerate}

\noindent
{\bf Case 1.} In the first case, $E$ does not examine $x^*$, which implies that $A$ wins if $y^{*\ell} = y^{*\ell}_0+y^{*\ell}_1 \in P_\ell$, and the event that $x^* \neq 1^\ell$ and $E$ returns $1^\ell$ does  not occur. We establish a lower bound on the probability that $y^{*\ell} \not \in P_\ell$ and derive an upper bound on the probability that $y^{*\ell} \in P_\ell$, which is an upper bound on  the probability that $A$ wins. Obviously, 
$$\mbox{Pr}[y^{*\ell} \not \in  P_\ell] \ge 
\mbox{Pr}[y^{*\ell} \not \in  P_\ell~|~\tilde{s}_b^{*(\ell)} \not \in T] \cdot \mbox{Pr}[\tilde{s}_b^{*(\ell)} \not \in T].$$

If $\tilde{s}_b^{*(\ell)} \not \in T$ then $\mbox{Pr}[y^{*\ell} \not \in  P_\ell]=\frac{|\G_\ell| - |P_\ell|}{|\G_\ell|}$ since $G$ is a random oracle and therefore $y^{*\ell}$ is randomly distributed in $\G_\ell$ and is independent of the view of $A$. 

To bound $\mbox{Pr}[\tilde{s}_b^{*(\ell)} \not \in T]$ we analyze the case that all the oracle  queries that  $(\IDPF.\Eval\Prefix)^G(b, k_0^*, \pp^*,x^*)$ makes in levels $j=i,\ldots,\ell$ are not in $T$. $(\IDPF.\Eval\Prefix)^G(b, k_0^*, \pp^*, x^*)$ makes two oracle calls in level $j$, one querying $s_b^{*(j-1)}$ and the second querying $\tilde{s}_b^{*(j)}$. Altogether, there are at most $2n$ such oracle calls, which we denote $q_1,\ldots,q_{2n}$ for convenience. Since $G$ is a random function and its values on $\{0,1\}^\lambda \setminus T$ are independent of the view of $A$, if $s_b^{*(j-1)} \not \in T$ then $\mbox{Pr}[\tilde{s}_b^{*(j)} \not \in T]=\frac{2^\lambda-t}{2^\lambda}$ and similarly, if $\tilde{s}_b^{*(j)} \not \in T$ then $\mbox{Pr}[s_b^{*(j)} \not \in T]=\frac{2^\lambda-t}{2^\lambda}$.  Therefore, $\mbox{Pr}[q_j \not \in T ~|~ q_1,\ldots,q_{j-1} \not \in  T] \ge \frac{2^\lambda-t}{2^\lambda}$. It follows that

\begin{eqnarray*}
\mbox{Pr}[s_b^{*(\ell)} \not \in T] & \ge & \mbox{Pr}[q_1,\ldots,q_{2n} \not \in T] \\
			& \ge & \left(1- \frac{t}{2^\lambda} \right)^{2n} \\
			& \ge & 1-\frac{2nt}{2^\lambda}
\end{eqnarray*}

and therefore, $\mbox{Pr}[y^{*\ell} \not \in P_\ell] \ge \left(1-\frac{|P_\ell|}{|\G_\ell|}\right)\left(1-\frac{2nt}{2^\lambda}\right)$ and
$$\mbox{Pr}[y^{*\ell} \in P_\ell] \le \frac{|P_\ell|}{|\G_\ell|}+\frac{2nt}{2^\lambda}.$$

\noindent
{\bf Case 2.} In the second case, $E$ does examine $x^*$, but outputs a different string of the same length $x^\ell$. That means that the $y$ values for both strings, $x^*$ and $x^\ell$ are in the valid set $P_\ell$. To bound the probability that this event occurs we consider  the construction of $\IDPF.\Eval\Prefix$.

There exist group elements $W^*_0, W^*_1,W_{CW} \in  \G_\ell$  such that:
$W^*_0+ W^*_1+W_{CW}=y^{*\ell}_0+y^{*\ell}_1 \in P_\ell$,
where $W^*_0$ is part of $G(\tilde{s}_0^{*\ell})$, $W^*_1$ is part of $G(\tilde{s}_1^{*\ell})$ and $W_{CW}$ is included in the public parameters $\pp$. It follows that there is a random element $W^*=W^*_0+W^*_1 \in \G_\ell$ and a part of the public parameters $W_{CW}$ such that $W^*+W_{CW} \in  P_\ell$. Using similar reasoning for the output value of the string $x^\ell$  it holds that there is a random  element $W^{(\ell)} \in \G_\ell$ such that $W^{(\ell)}+W_{CW} \in P_\ell$. Note that the same $W_{CW}$ is used for both strings since their length is equal.

Therefore, $W^*+W_{CW}-(W^{(\ell)}+W_{CW})=W^*-W^{(\ell)}$ is in $P_\ell^-$ and similarly $W^{(\ell)}-W^* \in P^-$. It follows that for the second case to occur, there must be an ordered pair of values among at most $t$ values that $A$ examines in the $\ell$-th level such that their difference is in $P_\ell^-$. The number of such ordered pairs is $t(t-1)$ and by union bound the probability that the difference between any pair is in $P^-$ is less than $\frac{|P^-_\ell|t^2}{|\G_\ell|}$.

\noindent
{\bf Case 3.} In the third case, all the oracle queries required to  evaluate $x^*$ are in $T$, but $E$ does not add $x^*$ to $C$ in its execution. The only way that could happen is that there are two different $i$-bit strings, $x^{*(i)}, x^{(i)}$ such that $s_b^{*(i)}=s_{b'}^{(i)}$, for $b,b' \in \{0,1\}$. In this situation, the If statement in Line \ref{ln:ifst} prevents the examination of $x^*$. 

Assuming that $x^{*(i)}$ is the first prefix of $x^*$ that triggers the If statement in Line \ref{ln:ifst}, there are two possible sub-cases. The first is that $\tilde{s}_b^{*(i)} \neq \tilde{s}_{b'}^{(i)}$, and the second is that $\tilde{s}_b^{*(i)} = \tilde{s}_{b'}^{(i)}$, but $s_b^{*(i-1)} \neq s_{b'}^{(i-1)}$. 

To bound the probability that $\tilde{s}_b^{*(i)} \neq \tilde{s}_{b'}^{(i)}$, but $s_b^{*(i)}=s_{b'}^{*(i)}$, note that $s_b^{*(i)}$ is a restriction of $G(\tilde{s}_b^{*(i)})$ to $\lambda$ bits, and $s_{b'}^{(i)}$ is a restriction of $G(\tilde{s}_{b'}^{(i)})$ to $\lambda$ bits. Since $G$ is a random function, the probability of finding two queries in $T$ that $G$ maps to the same output in $\{0,1\}^\lambda$ is at most $\binom{t}{2}/2^\lambda$. 

To bound the probability of the second sub-case, recall that either $\tilde{s}_b^{*(i)}=G(s_b^{*(i-1)})$ (abusing the notation to let $G(\cdot)$ denote the restriction of $G$ to the $\lambda$ bits used to derive $\tilde{s}_b$) or $\tilde{s}_b^{*(i)}=G(s_b^{*(i-1)}) \oplus s_{CW}$, for $s_{CW} \in \pp$. Similarly, either $\tilde{s}_{b'}^{(i)}=G(s_{b'}^{(i-1)})$ or $\tilde{s}_{b'}^{(i)}=G(s_{b'}^{(i-1)}) \oplus s_{CW}$.Thus, either $G(s_b^{*(i-1)})=G(s_{b'}^{(i-1)})$ or $G(s_b^{*(i-1)})=G(s_{b'}^{(i-1)}) \oplus s_{CW}$. The probability that either case occurs is at most $2\binom{t}{2}/2^\lambda$. 

By union bound, the probability that the third case occurs is at most $\frac{3\binom{t}{2}}{2^\lambda} \le \frac{3t^2}{2^\lambda}$.

taking a union bound over the three cases, we have that  
\begin{eqnarray*}
\epsilon_A	& \le	& \frac{|P_\ell|}{\G_\ell}+\frac{2nt}{2^\lambda}+\frac{|P^-_\ell|t^2}{|\G_\ell|}+\frac{3t^2}{2^\lambda} \\
						& \le & \frac{3t^2 + 2nt}{2^\lambda} + (t^2+1)\rho \\
						& \le & (4t^2+2nt+1)(\max \{\rho,1/2^\lambda\})
\end{eqnarray*}

\end{proof}

 }
\abbr{}{\section{Differential privacy details}
\label{app:dpx}

We describe how to modify the leakage profile of our
heavy-hitters protocol to satisfy a meaningful nation 
of differential privacy.
We claim no novelty of the technique we use to
provide differential privacy---similar ideas appear 
in prior work~\cite{DKMMN06,PrivEx,JJ16}
(see also the work of Balle et al.~\cite{BBGN20}). 
The main point here is that adding differential privacy to 
our system is simple.

\subsection{Implementing differential privacy}
The only information that the
protocol reveals to the servers about the clients' inputs
is the output of prefix-count oracle queries.
For the entire mechanism to provide
differential privacy, we need only
ensure that the outputs of these prefix-count
oracle queries satisfy differential privacy.

The presence or absence of a client's string in the
dataset can influence the value of any prefix-count oracle
query by $\pm 1$ at most. 
In this context, it is possible to achieve 
per-oracle-query $\epsilon$-differential privacy using the
Laplace mechanism~\cite{DR14}.
Specifically, the prefix-count oracle 
samples a noise value from the Laplace distribution with parameter $1/\epsilon$
and masking the oracle's output with this noise.

To implement this noise in the two-server setting
(in which one of the two servers may be malicious),
each of the two servers can sample and add these
noise values independently.
More specifically, in \cref{step:oracle} of \cref{proto:heavy}, 
when server $b \in \zo$ processes a prefix-count oracle query on 
prefix $p \in \zo^*$, the server also samples a noise value
$\nu_{p,b} \getsr \mathsf{Laplace}(1/\epsilon)$.
The server then publishes the noised value, rounded to the nearest integer:
$\val_{p,b}' \gets \val_{p,b} + \mathsf{Round}(\nu_{p,b}) \in \Z$.

When the servers run \cref{proto:heavy} looking for $t$-heavy hitters
(i.e., with heaviness threshold $t$) on strings of length $n$
with $C$ total clients, the total number of prefix-count oracle
queries they make is $q = n\cdot C/t$.
Applying the advanced composition theorem~\cite{DR14} for differential 
privacy, we find that if the per-query privacy parameter is $\epsilon$,
then for any $\delta' > 0$, the entire output of \cref{proto:heavy} satisfies 
$(\epsilon', \delta')$-differential privacy, where
$\epsilon' = \sqrt{2q\ln(1/\delta')}\cdot \epsilon + q \epsilon(e^\epsilon - 1)$.

In \abbr{the full version of this work~\cite{full}}{\cref{app:dpnoise}}, we calculate how many clients a deployment will
need to ensure that (with good probability) the differential-privacy noise 
will not change the server's view of which strings are the heavy hitters.
In \abbr{the full version of this work~\cite{full}}{\cref{app:dpfull}}, we 
give an example derivation of the differential-privacy parameters.

\subsection{Noise analysis}
\label{app:dpnoise}

Providing differential privacy inherently introduce some noise
into the protocol's output. 
We can, however, bound the probability that the noise is 
A standard tail bound on the Laplace distribution shows
that, for any $\lambda \geq 1$, the probability that the noise
that the servers collectively add has magnitude more than $2\lambda/\epsilon$
is at most $\exp(-\lambda)$ for a single query.
Applying a union bound across all $q$ queries shows that the chance
of a large deviation is then at most $q\exp(-\lambda)$.

If the deviation $2\lambda/\epsilon$ is much smaller than the
heavy-hitters threshold $t$, then extra noise will not cause 
correctness failures---false negatives (heavy hitters that the
servers do not output) or false positives (non-heavy hitters that
the servers do output).
If we take $2\lambda/\epsilon < 0.05t$, for example, then the noise
will never shift the weight on any potential heavy-hitter by 
more than $\pm 0.05t$.

If the heavy-hitters threshold $t = 0.01C$, for $C$ clients,
then choosing $\epsilon$ such that $4000\lambda/C < \epsilon$
will guarantee that servers' weight on a given string never 
deviates by more than $0.05t$.
So the servers will output all strings that at least $1.05t$ clients
hold and will not output any strings that fewer than $0.95t$ clients hold.

\subsection{Example parameter setting}
\label{app:dpfull}

A company that deploys our system must choose:
\begin{itemize}
\item the heavy-hitters threshold $t$,
\item the desired privacy parameters $(\epsilon', \delta')$, and
\item the maximum tolerable correctness-failure probability.
\end{itemize}

\begin{figure}
  \centering
  \includegraphics{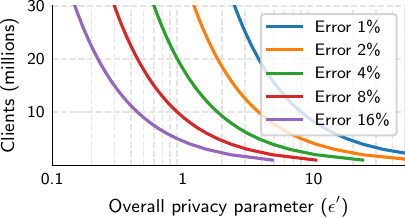}
  \caption{For a given privacy parameter $\epsilon'$, and a
          given correctness error (i.e., the percentage difference
          between the true weight on a string and the server-computed
          weight), the graph shows the minimum number of clients
          needed to participate. 
          We take the client's string length $n=256$, and we fix
          the private error $\delta' = 2^{-40}$.}
          \label{fig:budget}
\end{figure}

Concretely, when the servers search 
for heavy-hitters that at least 1\% of the
clients hold, we have $t = 0.01C$ and $q = 100n$.
If we allow differential privacy to fail to hold
with probability at most $\delta' = 2^{-40}$
and the clients hold strings of length $n = 256$,
then if the output of each oracle query satisfies 
$\epsilon$-differential privacy with $\epsilon = 0.001$,
the overall protocol output satisfies
$(\epsilon',\delta')$-differential privacy with
$(\epsilon',\delta') = (1.22, 2^{-40})$.

With this parameter setting, 
the two servers will collectively
add noise from the Laplace distribution 
with zero mean and parameter $2/\epsilon = 2000$.
So, if we accept a correctness failure one in a billion
protocol runs, we can take $\kappa = 30$ and the 
per-query noise will be bounded by $\pm 60,000$.

In a deployment with 50 million clients, a string is a $1\%$
heavy hitter if more than 500,000 clients hold the string.
So, the protocol will, with overwhelming probability, output
strings that more than $560,000$ clients hold and will not
output strings that fewer than $440,000$ clients hold.
\cref{fig:budget} shows how the minimum number of users needed
changes as a function of the privacy budget $\epsilon'$ and
the correctness error (i.e., the difference between a string's true 
weight and the weight that the servers compute for it).
 }

\end{document}